\documentclass[aps,prb,twocolumn,superscriptaddress]{revtex4-2}
\usepackage{amsmath}    
\newcommand{\CP}{\mathbb{CP}}
\usepackage{amsfonts}
\usepackage{amsthm}
\usepackage{amssymb}
\usepackage[colorlinks=true,citecolor=blue,linkcolor=red]{hyperref}
\usepackage{graphicx}
\usepackage{bbold}					% for \mathbb{1} as a nice unit matrix symbol
\usepackage[dvipsnames]{xcolor}
\usepackage{xcolor} % allows coloring text
\usepackage{bm} % for bold math symbols

\usepackage{times}
\usepackage{bm}

\usepackage{dcolumn}
\usepackage{mathtools}

\usepackage{array,booktabs,ragged2e}	
\usepackage{soul}						% strike out text using \st{}
\usepackage{cancel}						% detto, but more clearly visible when reading
\usepackage{tabularx}% for table in the CHS-model discussion
\usepackage{multirow}
\usepackage{hhline}
\usepackage{adjustbox}
\usepackage{sidecap}
\usepackage{threeparttable}
\usepackage{colortbl} %added by JX, for cellcolor in tables.
\usepackage{tabularx}
\usepackage[utf8]{inputenc}
\usepackage[T1]{fontenc}

\usepackage{amsmath, amssymb, bm}

\newcommand{\eq}[1]{\begin{equation}#1\end{equation}}
\newcommand{\BZ}{\ensuremath{\mathrm{BZ}}}   % Brillouin zone
\DeclareMathOperator{\Tr}{Tr}   % trace operator

\def\bra#1{\left<{#1}\right|}					% "bra"-state
\def\ket#1{\left|{#1}\right>}					% "ket"-state
\def\braket#1#2{\left<{#1}|{#2}\right>}			% "bra-ket"-product
\newcolumntype{P}[1]{>{\raggedleft\arraybackslash}p{#1}}
\newcolumntype{R}[1]{>{\centering\arraybackslash}p{#1}}

\usepackage{bm,color,amsmath,amssymb,mathrsfs,latexsym,graphicx,psfrag,mathtools}
\usepackage{color}
\usepackage[dvipsnames]{xcolor}

\usepackage{empheq}%\begin{empheq}[box=\fbox]{align} ... \end{empheq}

\newcommand{\bsl}[1]{{\boldsymbol{#1}}}

\newcommand{\ii}{\mathrm{i}}
\newcommand{\dd}{\mathrm{d}}
\newcommand{\vol}{\mathop{\mathrm{vol}}}

\newcommand{\dsZ}{\mathbb{Z}}

\newcommand{\dsR}{\mathbb{R}}

\newcommand{\dsC}{\mathbb{C}}

\newcommand{\Gr}{\mathop{\mathrm{Gr}}}

\renewcommand{\Im}{\mathop{\mathrm{Im}}}

\newcommand{\refcite}[1]{Ref.\cite{#1}}

\newcommand{\mat}[1]{\left(\begin{matrix}#1\end{matrix}\right)}

\newcommand{\eqa}[1]{\begin{align}\begin{split} #1 \end{split}\end{align}}

\let\oldAA\AA
\renewcommand{\AA}{\text{\normalfont\oldAA}}

\newcommand{\ie}{{\emph{i.e.}}}
\newcommand{\eg}{{\emph{e.g.}}}

\newcommand{\diag}{\text{diag}}

\newcommand{\Ch}{\text{Ch}}

\newtheorem{theorem}{Theorem}

\usepackage{comment}

\usepackage{cleveref}
\crefname{appendix}{App.}{Apps.}
\Crefname{appendix}{Appendix}{Appendices}
\crefname{equation}{Eq.}{Eqs.}
\Crefname{equation}{Equation}{Equations}
\Crefname{figure}{Figure}{Figures}
\crefname{figure}{Fig.}{Figs.}
\crefname{theorem}{Thm.}{Thms.}
\Crefname{theorem}{Theorem}{Theorems}
\crefname{table}{Tab.}{Tabs.}
\Crefname{section}{Section}{Sections}
\crefname{section}{Sec.}{Secs.}
\creflabelformat{appendix}{#2#1#3}

\newcommand{\norm}[1]{\left\lVert#1\right\rVert}

\begin{document}

\title{Ideal Bands in Tight-Binding Models}

\author{Lexu Zhao}
\thanks{These authors contributed equally to this work.}
\affiliation{Department of Physics and Quantum Theory Project, University of Florida, Gainesville, Florida 32611, USA}

\author{Yang Ge}
\thanks{These authors contributed equally to this work.}
\affiliation{Department of Physics and Quantum Theory Project, University of Florida, Gainesville, Florida 32611, USA}

\author{Shinsei Ryu}
\affiliation{Department of Physics, Princeton University, Princeton, New Jersey 08544, USA}

\author{Jiabin Yu}
\email{yujiabin@ufl.edu}
\affiliation{Department of Physics and Quantum Theory Project, University of Florida, Gainesville, Florida 32611, USA}

\begin{abstract}
A band is called ideal when its Dirichlet functional saturates the topological lower bound.
We study ideal bands in finite-band tight-binding models with conventional two-dimensional lattice translation symmetries, allowing the bands to have non-flat dispersion.
We first provide an analytic construction of isolated Chern-ideal bands with Chern number $|\Ch|=1$ in finite-band models with exponentially decaying hopping.
This construction applies only when at least two orbitals have different embedded positions (modulo lattice vectors), complementing the previously known construction for $|\Ch|>1$.
We then show that isolated Chern-ideal bands with any nonzero Chern number cannot exist in finite-band models with finite-range hopping, regardless of the embedded orbital positions.
The conclusion holds even if there are isolated band touching points, as long as the Berry curvature does not diverge anywhere in the Brillouin zone.
We finally generalize the conclusions to Wilson-loop-ideal bands with zero total Chern number, such as Kane-Mele $\dsZ_2$-ideal bands.
\end{abstract}

\maketitle

\section{Introduction}
The quantum metric is a classic quantum-geometric quantity that captures how fast the periodic part of Bloch states changes in the first Brillouin zone (BZ)~\cite{Liu_2024,Yu2025QGReview,gao2025quantumgeometryphenomenacondensed,Jiang_2025,Verma_2026,kitamura2026quantumgeometrycorrelatedelectron,Peotta2025QGSuperfluidity,cano2026idealquantumgeometryfractional,Yang_2025_QGMetrology}.
When the Dirichlet functional of a band---which, in our setting, reduces to the integral of the trace of the quantum metric---saturates the topological lower bound, the corresponding Bloch states acquire special properties~\cite{Laughlin1983FQHE,Haldane1983FQH,Jian_2012,Roy2014QGChernBound,claassen2015position,Lee_2017,Ledwith_2020,Ozawa_2021,Mera_2021,Jie2021IdealBands,Mera2021FlatBandsKahler,Wang_2022,Ledwith_2022,northe2022interplay,Wang_2023,Valentin2023IdealBands,Parker2023IdealBands,Dong_2023_ideal_Higher_Chern,BJY2024EulerBoundQG,liu2024theorygeneralizedlandaulevels,ji2024quantum,Okuma_2024_VortexFunctions,Roy2024FCINonLL,Queiroz_2024_TMD_Chiral,Oreg2025UnifyingFrameworkFCI,Yu_2025_WL_ideal,Sun2025IdealChern,Liu_2025_idealTHF,Fukui_2025,Fujimoto_2025}.
The most classic example is the lowest Landau level, on which Laughlin constructed an analytical wavefunction~\cite{Laughlin1983FQHE} for fractional quantum Hall~\cite{Tsui1982FQHE} states through vortex attachment, \ie, by multiplying factors of $(z_i-z_j)$, where $z_i$ is the complex coordinate of the $i$th particle.
This many-body construction was later generalized to ideal bands, or, more precisely, Chern-ideal bands~\cite{Jie2021IdealBands,Ledwith_2022,Wang_2022,Parker2023IdealBands,Valentin2023IdealBands,Wang_2023,Dong_2023_ideal_Higher_Chern}---bands with a quantum metric saturating the lower bound set by the Chern number~\cite{Roy2014QGChernBound,Bellissard1994QGChernBound}, providing the construction of the wavefunction for fractional Chern insulators (FCIs)~\cite{neupert, sheng, regnault,Sun2011,Tang11}.
Recently, \refcite{Yu_2025_WL_ideal} further generalized the concept of ideality to Wilson-loop-ideal bands, whose quantum metric saturates the Wilson-loop (WL) lower bound~\cite{Yu2025Z2bound}.
Two examples of WL-ideal bands with zero total Chern number are the ideal bands with a nontrivial time-reversal-protected Kane-Mele $\dsZ_2$ topology~\cite{Kane2005Z2,Zhang2006QSH,Kane2005QSH,Bernevig2006BHZ} or nontrivial inversion-protected fragile topology~\cite{song2020}.
Both of them can be unitarily transformed into a pair of Chern-ideal states with opposite Chern numbers. Attaching vortices to them produces a pair of FCIs with opposite Hall conductances (such as fractional topological insulators~\cite{Zhang2006QSH,Levin_Stern,Neupert2011FTI,Stern2015review,Neupert_2015} in the $\mathbb{Z}_2$-ideal case), which is still topologically ordered as FCIs are non-invertible~\cite{Cheng_2026_Ordering_TopoOrder}.

Previous studies of ideal bands have mainly focused on continuum settings, where the underlying models contain an infinite number of bands. 
In this work, we focus on ideal bands in tight-binding models with a finite number of bands (\ie, a finite number of degrees of freedom per unit cell).
Note that we do not require the bands to have flat dispersions.
We also only require the models to have ordinary lattice translations rather than magnetic translations, thus differing from the Kapit-Mueller construction~\cite{Kapit_2010}.
Previously, it was shown that when all orbitals have the same embedded positions (modulo lattice vectors), one can only construct the isolated Chern-ideal bands with Chern numbers $|\Ch|>1$ in a finite-band model with exponentially decaying hopping~\cite{Jian_2012,claassen2015position, Lee_2017,Mera_2021}.
Then, is it possible to construct isolated $|\Ch|=1$ Chern-ideal bands in a finite-band model with exponentially decaying hopping when not all orbitals have the same embedded positions (modulo lattice vectors)?
Furthermore, is it possible to construct a Chern-ideal band with any nonzero Chern number in a finite-band model with finite-range hopping?
The corresponding questions also remain unanswered for other WL-ideal bands.
 
In this work, we provide the answers to these two questions.
We first show that, as long as at least two orbitals have different embedded positions (modulo lattice vectors), one can explicitly construct $|\Ch|=1$ Chern-ideal bands in finite-band models with exponentially decaying hopping, using Jacobi theta functions.
Then we find that, regardless of the embedded positions, isolated Chern-ideal bands with nonzero Chern number are not allowed in any finite-band model with finite-range hopping.
Even if the Chern-ideal band is allowed to touch other bands at finitely many isolated momentum points, it still does not exist in finite-band models with finite-range hopping, as long as the Berry curvature does not diverge (\eg, the critical topological bands proposed in \refcite{li2026stabletopologyexactlyflat,liu2026theorytopologicaltopologicalflatbands}). 
The conclusions are directly applied to any isolated set of two WL-ideal bands with zero total Chern number, nontrivial normal WL winding, and non-singular non-Abelian Berry curvature, owing to their unitary equivalence to Chern-ideal bands.
Here ``normal'' means that the WL path is along a primitive reciprocal lattice vector.

\section{Setup}
We first present the setup on which our analysis will be built.
We consider a two-dimensional (2D) Bravais lattice $\mathcal L$ with
lattice vectors $\bsl R\in\mathcal L$ and $N$ internal degrees of
freedom per unit cell, labeled $\alpha=1,\ldots,N$.
We refer to these internal degrees of freedom collectively as
``orbitals'', although they may also represent spins, valleys, or sublattices. 
The real-space basis states are denoted by
$\ket{\bsl R,\alpha}$, and the physical position of orbital $\alpha$
in the unit cell $\bsl R$ is
$\bsl R+\bsl\tau_\alpha$.
The tight-binding Hamiltonian we consider has the conventional
lattice-translation symmetry, and thus has the general form:
\eq{
H = \sum_{\bsl R,\bsl R'\in\mathcal L}
\sum_{\alpha,\alpha'=1}^{N}
\left[
t_{\bsl R-\bsl R'}
\right]_{\alpha\alpha'}
\ket{\bsl R,\alpha}
\bra{\bsl R',\alpha'}.
\label{eq:real_space_hamiltonian_main}
}
Here
$\left[t_{\bsl R-\bsl R'}\right]_{\alpha\alpha'}$
is the hopping amplitude from orbital $\alpha'$ in unit cell
$\bsl R'$ to orbital $\alpha$ in unit cell $\bsl R$. To define the momentum-space orbital basis, we introduce embedded
positions $\{\bsl\xi_\alpha\}$ and set
\eq{
\ket{\bsl k,\alpha}
=
\frac{1}{\sqrt{N_k}}
\sum_{\bsl R\in\mathcal L}
e^{\ii
\left(
\bsl R+\bsl\xi_\alpha
\right)\cdot\bsl k}
\ket{\bsl R,\alpha},
\label{eq:embedded_fourier_transform_main}
}
where $N_k$ is the number of unit cells.
The periodic and physical embeddings correspond to the choices $\bsl{\xi}_\alpha=\bsl{0}$ and $\bsl{\xi}_\alpha=\bsl{\tau}_\alpha$, respectively \cite{haldane2014attachmentsurfacefermiarcs,Lim_2015,Simon_2020,Fuchs_2021,Huhtinen_2022,PhysRevLett.133.246603}.
Our discussion is formulated for a general embedding $\bsl{\xi}_\alpha$, and therefore automatically includes both cases. 
In this basis, the $N\times N$ Bloch Hamiltonian has matrix elements
\eq{
\left[
h(\bsl k)
\right]_{\alpha\alpha'}
=
\sum_{\bsl R\in\mathcal L}
\left[
t_{\bsl R}
\right]_{\alpha\alpha'}
e^{-\ii\bsl k\cdot
\left(
\bsl R
+
\bsl\xi_\alpha
-
\bsl\xi_{\alpha'}
\right)}.
\label{eq: bloch_hamiltonian_embedding_main}
}
Thus, the Fourier phase associated with the hopping element
$\left[
t_{\bsl R}
\right]_{\alpha\alpha'}$ is determined by the embedded
displacement $\bsl\rho
=
\bsl R
+
\bsl\xi_\alpha
-
\bsl\xi_{\alpha'}$. 
For a reciprocal-lattice vector $\bsl G$, we have
\eq{
\ket{\bsl k+\bsl G,\alpha}
=
\ket{\bsl k,\alpha} e^{\ii\bsl G\cdot\bsl\xi_\alpha}.
}
Consequently, the domain of $h$ is exactly extended to the whole $\dsR^2$ with 
\eq{
h(\bsl k+\bsl G) =
V_{\bsl G}
h(\bsl k)
V_{\bsl G}^{\dagger}
\label{eq:bloch_hamiltonian_covariance_main}
}
with $V_{\bsl G}
= \operatorname{diag}
\left(
e^{-\ii\bsl G\cdot\bsl\xi_1},
\ldots,
e^{-\ii\bsl G\cdot\bsl\xi_N}
\right)$.

Given a band $E_{\bsl{k}}$, we define its projector as
\eq{
P_{\bsl{k}} = U_{\bsl{k}} U^\dagger_{\bsl{k}}
}
with $U_{\bsl{k}}$ the eigenvector for $E_{\bsl{k}}$
\eq{
h(\bsl{k})U_{\bsl{k}} = E_{\bsl{k}} U_{\bsl{k}}\ .
}
The quantum metric~\cite{Fubini1904,Study1905} and Berry curvature~\cite{TKNN} are defined, respectively, as 
\eqa{
& \left[g_{\bsl{k}}\right]_{ij}
=
\frac{1}{2}\Tr\left[
\partial_{k_i}P_{\bsl{k}}
\partial_{k_j}P_{\bsl{k}}
\right]\ ,
\\
& F_{\bsl{k}}
=
\ii\Tr\left[
P_{\bsl{k}}
\left[
\partial_{k_x}P_{\bsl{k}},
\partial_{k_y}P_{\bsl{k}}
\right]
\right]\ .
\label{eq:quantum_metric_berry_curvature}
}
We always choose $E_{\bsl{k}} = E_{\bsl{k}+\bsl{G}}$ with 
\eq{
\label{eq:physical_boundary_main}
U_{\bsl{k}+\bsl{G}} = V_{\bsl{G}}U_{\bsl{k}} e^{\ii \phi_{\bsl{k},\bsl{G}}}
}
with $\phi_{\bsl{k},\bsl{G}}\in \dsR$, leading to 
\eq{
P_{\bsl{k}+\bsl{G}} = V_{\bsl{G}}P_{\bsl{k}}V_{\bsl{G}}^\dagger \ ,\ g_{\bsl{k}+\bsl{G}} = g_{\bsl{k}} \ ,\ F_{\bsl{k}+\bsl{G}} = F_{\bsl{k}} \ .
}
The condition is naturally satisfied when $E_{\bsl{k}}$ is isolated throughout the $\BZ$.

A band $E_{\bsl{k}}$ is called Chern-ideal \cite{Parker2023IdealBands,Jie2021IdealBands}
if 
\eq{
\frac{1}{2\pi}\int_{\BZ} d^2 k \Tr[g_{\bsl{k}}] = |\Ch| \equiv \left| \frac{1}{2\pi}\int_{\BZ} d^2 k F_{\bsl{k}}\right|\ ,
}
where we note that $\int_{\BZ} d^2 k \Tr[g_{\bsl{k}}]$ is the simplified form of the Dirichlet functional associated with the flat metric on the BZ torus inherited from $\dsR^2$~\cite{Yu_2025_WL_ideal}. (See also \cref{app:quantum-geometry}.)
As $\Ch<0$ can be related to $\Ch>0$ by a time-reversal operation on the wavefunction, we henceforth only consider 
\eq{
\Ch>0\ .
}
In this case, the Chern-ideal condition is equivalent to
\eq{
\Tr[g_{\bsl{k}}] = F_{\bsl{k}}\ \quad \forall\bsl{k}\in\BZ,
}
as $\Tr[g_{\bsl{k}}]$ is always nonnegative. 

When the embedded positions $\bsl{\xi}_{\alpha}$ change, the integrated quantum metric can change (especially when the shift is not uniform for all orbitals) while the Chern number remains unchanged.
Therefore, a band can be ideal in one embedding but not ideal in another, as discussed in the next section.

\section{Chern-Ideal Bands with $\Ch=1$ in Models with Exponentially Decaying Hopping}
\label{sec:Chern-Ideal Bloch States in the Finite-Orbital Hilbert Space}

We first discuss isolated Chern-ideal bands in tight-binding models with exponentially decaying hopping, $\ie$, elements of $t_{\bsl{R}}$ in \cref{eq: bloch_hamiltonian_embedding_main} have amplitudes decaying exponentially in $|\bsl{R}|$.
In this type of models, the projector $P_{\bsl{k}}$ of the isolated band $E_{\bsl{k}}$ is real-analytic in $\bsl{k}$ over the whole $\dsR^2$ \cite{Kuchment_2008,Kuchment_2016} (see \cref{app: basics} for details).
As a result, the quantum metric $g_{\bsl{k}}$ and the Berry curvature $F_{\bsl{k}}$ are real-analytic over the whole $\dsR^2$.

Given an isolated Chern-ideal band, its projector satisfies
\cite{Jian_2012,claassen2015position,Ledwith_2020,Jie2021IdealBands,Mera_2021,Ledwith_2022,Parker2023IdealBands}
\eq{
(1-P_{\bsl{k}})(\partial_{k_x} - \ii \partial_{k_y}) P_{\bsl{k}}  = 0\ ,
\label{eq:chern-ideal-projector}
}
leading to 
\eq{
(\partial_{k_x} - \ii \partial_{k_y}) U_{\bsl{k}} = \lambda_{\bsl{k}} U_{\bsl{k}}
\label{eq:chern-ideal-Uk-lambda_k}
}
with $\lambda_{\bsl{k}}$ a nonzero scalar.
Then, let $U_{\bsl{k},a}$ be a nonzero component of $U_{\bsl{k}}$ at $\bsl{k}$, we have
\eq{
(\partial_{k_x} - \ii \partial_{k_y}) \left[\frac{U_{\bsl{k}}}{U_{\bsl{k},a}}\right] =0\ 
\label{eq:Uk_over_Uka}
}
at $\bsl{k}$ (more precisely, in a small neighborhood of $\bsl{k}$).
Let $k=k_x-\ii k_y$ and $w(k)=U_{\bsl{k}}/U_{\bsl{k},a}$. Such a $w(k)$ is holomorphic at $\bsl{k}$~\cite{Jian_2012,claassen2015position,Lee_2017,Mera_2021,Parker2023IdealBands,cano2026idealquantumgeometryfractional}.
Therefore, the eigenvector is locally holomorphic (up to a factor) at any $\bsl{k}$.

In fact, we can fix the choice of the component $a$ for all $\bsl{k}\in \dsR^2$, and the resulting $w(k)$ is holomorphic everywhere except where $U_{\bsl{k},a}=0$.
Owing to the local holomorphicity (up to a factor) everywhere in the $\BZ$, the zeros of $U_{\bsl{k},a}$ can only be isolated points.
(See \cref{sec:exp} for details.)
Therefore, the constructed $w(k)$ is meromorphic throughout $\dsC$, and the construction of the ansatz is eventually reduced to finding the meromorphic $w(k)$.
In the end, the construction will produce a real-analytic projector
$P_{\bsl{k}}$.
Once such a projector is constructed, it defines a parent tight-binding model with exponentially decaying hopping, as discussed in \cref{app: basics}.

We now provide an explicit construction of an isolated Chern-ideal band with $\Ch=1$.
Consider a two-band model with two orbitals, $A$ and $B$, at different embedded positions, such that $\Delta\boldsymbol\xi\equiv\boldsymbol\xi_B-\boldsymbol\xi_A\notin\mathcal L.$ 
As discussed above, we can define the following meromorphic vector
\eq{w(k)=\mat{ 1 \\ \frac{U_{\bsl{k},B}}{U_{\bsl{k},A}} }\equiv\mat{ 1 \\ w_B(k) }. \label{eq:w_k}}
Based on \cref{eq:physical_boundary_main}, $w(k)$ observes the condition in momentum space
\eq{
w(k+G_x-\ii G_y) = \diag(1,e^{-i\bsl G\cdot\Delta\boldsymbol\xi})w(k).
\label{eq:WB_quasiperiodic}
}
To simplify this relation, we introduce the normalized complex momentum coordinate $z=\frac{k_x-ik_y}{b_{1x}-ib_{1y}}$ and $\omega=\frac{b_{2x}-ib_{2y}}{b_{1x}-ib_{1y}}$, where $\bsl{b}_1$ and $\bsl{b}_2$ are the primitive reciprocal lattice vectors.
We choose $\bsl{b}_1$ and $\bsl{b}_2$ such that $\operatorname{Im}\omega>0$, as such choices always exist (\cref{subsub: construction for f at diff positions}).
Then, we can define
\eq{
\label{eq:f_ratio_main}
f(z) = w_B[z (b_{1x}-ib_{1y})]\ ,
}
which simplifies \cref{eq:WB_quasiperiodic} to 
\eq{
f(z+1)=e^{-2\pi i\nu_1}f(z),
\quad f(z+\omega)=e^{-2\pi i\nu_2}f(z)\ ,
\label{eq:f_quasiperiodic}
}
where $\Delta\boldsymbol\xi=\nu_1\bsl a_1+\nu_2\bsl a_2(\mathrm{mod} \mathcal L)$ with $\nu_1,\nu_2\in[0,1)$ and $(\nu_1,\nu_2)\neq(0,0)$, and $\bsl{a}_1$ and $\bsl{a}_2$ are two corresponding primitive lattice vectors. The desired $f(z)$ can be constructed by the odd Jacobi theta function written in \cite{Haldane_1985_LJ_function,Kharchev_2015,Slavnov_2020,kulkarni2023formfactorlocaloperators}:
\eq{\Theta(z|\omega)
 =-i\sum_{n\in\mathbb Z}(-1)^n
 e^{i\pi\omega(n+\frac12)^2}
 e^{2\pi i(n+\frac12)z}, \label{eq:theta_definition_main}}
which obeys
\eqa{
\Theta(z+1\mid\omega)
&=-\Theta(z\mid\omega),\\
\Theta(z+\omega\mid\omega)
&=-e^{-\pi i\omega-2\pi iz}\Theta(z\mid\omega).
\label{eq:theta_quasiperiodicity}
}
Then, for any $p\in\dsC$, we define:
\eq{f(z)=e^{-2\pi i\nu_1 z}
 \frac{\Theta\left(z - p-\nu_1\omega+\nu_2 \mid \omega\right)}
 {\Theta\left(z- p \mid \omega\right)} \ ,
 \label{eq:explicit_Ch1_f_main}}
which satisfies the desired conditions, as follows directly from \cref{eq:theta_quasiperiodicity}.
$\Theta(z \mid \omega)$ is entire and has a lattice of simple zeros at $\mathbb Z+\omega\mathbb Z$, as reviewed in \cref{subsub: construction for f at diff positions}.  
Hence the constructed $f(z)$ has simple poles at $z=p+m+\omega n$ and simple zeros at $z=p+m+\omega n+\nu_1\omega-\nu_2,$ with $n,m\in\dsZ$.
From the constructed $f(z)$ in \cref{eq:explicit_Ch1_f_main}, we can get $w(k)$, $U_{\bsl{k}}$ and $P_{\bsl{k}}$ as
\eqa{
& w(k) = \mat{ 1 \\ f(z) } ,\\
& U_{\bsl{k}} = \frac{1}{\sqrt{1+|f(z)|^2}}\mat{ 1 \\ f(z) } ,\\
& P_{\bsl{k}} = \frac{1}{1+|f(z)|^2}\mat{ 1 & f^*(z) \\ f(z) & |f(z)|^2 }.}

As verified in
\cref{subsub: construction for f at diff positions}, the constructed $P_\bsl{k}$ satisfies several properties. First, it is real-analytic over $\dsR^2$. Second, its Chern number is $\Ch = 1$, because $f(z)$ only has one simple pole in the $\BZ$.
Third, $P_{\bsl{k}}$ satisfies $\Tr[g_{\bsl{k}}] = F_{\bsl{k}}$ for any $\bsl{k}$.
Thus, $P_{\bsl{k}}$ is the projector of a Chern-ideal band with $\Ch=1$ in a two-band tight-binding model with exponentially decaying hopping.
Note that $\Delta\bsl{\xi}\notin\mathcal L$ is crucial to this construction.
If $\Delta\bsl{\xi}\in\mathcal L$, then
$\nu_1=\nu_2=0$, and \cref{eq:explicit_Ch1_f_main} reduces to
$f(z)=1$. The resulting projector is constant and therefore has Chern number zero.
On the other hand, $p$ is a free parameter here---any choice of $p$ gives a correct construction of the Chern-ideal band with $\Ch=1$.
Moreover, $\Tr[g_{\boldsymbol{k}}]$ vanishes at the zeros of $\partial_zf$. Since $\partial_zf$ has one double pole per BZ and
obeys the same transformation rule under the shift of reciprocal lattice vectors as $f$, it has exactly two zeros in the BZ, counted with multiplicity, as shown in \cref{subsub: construction for f at diff positions}.
This is consistent with the enforced zeros of the Berry curvature of an isolated band in a two-band model \cite{Varjas_2022}.

\section{No Chern-Ideal Bands in Models with Finite-Range Hopping}\label{sec:Chern-Ideal Bands in Tight-Binding Model}

The previous section (\cref{sec:Chern-Ideal Bloch States in the Finite-Orbital Hilbert Space}) studied Chern-ideal bands in models with exponentially decaying hopping.
We now first show that strictly finite-range hopping cannot realize an isolated Chern-ideal band, regardless of the number of orbitals and the embedding.
Then, we will move onto the critical bands.

Finite-range hopping means that there are only finitely many $\bsl{R}$'s for which $t_\bsl{R}\ne0$.
Let us define the finite set of all distinct  embedded displacements of nonzero hoppings:
\eq{
\mathcal{D}
=
\left\{
\bsl{R}
+
\bsl{\xi}_{\alpha}
- \bsl{\xi}_{\alpha'}
\mid
\bsl{R}\in\mathcal{L},
\alpha,\alpha'=1,\ldots, N,
\left[t_\bsl{R}\right]_{\alpha\alpha'}\neq 0
\right\}.
}
We note that different orbital pairs $(\alpha,\alpha')$ and Bravais-lattice vectors $\bsl{R}$ may give the same displacement $\bsl{\rho} = \bsl{R}+\bsl{\xi}_{\alpha}-\bsl{\xi}_{\alpha'}$.
The set $\mathcal{D}$ contains only one element for each distinct displacement.
With the set of $\bsl{\rho}$ defined, we can rewrite the Fourier transformation of the Hamiltonian in \cref{eq: bloch_hamiltonian_embedding_main} as
\eq{
h(\bsl{k})
=
\sum_{\bsl{\rho}\in\mathcal{D}}
e^{-i\bsl{k}\cdot\bsl{\rho}}
A_{\bsl{\rho}},
\label{eq:finite_range_physical_sum_main}
}
where $\left[A_{\bsl{\rho}}\right]_{\alpha\alpha'} = 0 $ if there is no hopping with displacement $\bsl{\rho}$ from the orbital $\alpha'$ to $\alpha$, and $\left[A_{\bsl{\rho}}\right]_{\alpha\alpha'} = \left[t_{\bsl{R}}\right]_{\alpha\alpha'}$ with $\bsl{\rho}=\bsl{R}+\bsl{\xi}_{\alpha}-\bsl{\xi}_{\alpha'}$ otherwise.
\Cref{eq:finite_range_physical_sum_main} holds because for a specific $\bsl{\rho}\in \mathcal{D}$ and a pair of orbitals $(\alpha,\alpha')$, there is either zero or one $\bsl{R}$ that satisfies $\bsl{\rho} = \bsl{R}+\bsl{\xi}_{\alpha}-\bsl{\xi}_{\alpha'}$.

Now suppose there exists
an isolated Chern-ideal band with dispersion $E_{\bsl{k}}$ in a finite-range-hopping model.
Since $\Ch>0$, there must exist $\bsl{k}_0$ such that $\Tr[g_{\bsl{k}_0}]>0$.
At such a $\bsl{k}_0$, we have the holomorphic $w(k)$ in \cref{eq:w_k} defined in a small enough open \emph{square} neighborhood of $\bsl{k}_0$, \ie, $k_x \in I_x \equiv (-\epsilon+k_{0,x}, \epsilon+k_{0,x})$ and $k_y \in I_y \equiv (-\epsilon+k_{0,y}, \epsilon+k_{0,y})$, for some small positive $\epsilon$.
Throughout the neighborhood $I_x \times I_y$, $w(k)$ satisfies 
\eq{h(\bsl k) w(k)=E_{\bsl{k}}w(k), \label{eq:Hw_main_proof}}
which leads to
\eq{
\left[h(\bsl k) w(k)\right]_a = E_{\bsl{k}}
}
where $a$ is the index of the chosen nonzero component of $U_{\bsl{k}}$ in \cref{eq:Uk_over_Uka}.
For any $\alpha \neq a$, we then have 
\eq{\sum_{\boldsymbol\rho\in\mathcal D}
e^{-i\bsl k\cdot\boldsymbol\rho}
C_{\alpha,\boldsymbol\rho}(k)=0
\label{eq:C_sum_main}}
for all $\bsl{k}$ in the neighborhood $I_x \times I_y$, where
\eq{C_{\alpha,\boldsymbol\rho}(k) =[A_{\boldsymbol\rho} w(k)]_\alpha -w_\alpha(k)[A_{\boldsymbol\rho} w(k)]_a. \label{eq:C_definition_main}} 
Owing to the fact that each $C_{\alpha,\boldsymbol\rho}$ is holomorphic in $k$, \cref{eq:C_sum_main} implies 
\eq{C_{\alpha,\boldsymbol \rho}(k)=0, \quad \forall \alpha\neq a,\bsl{\rho} \in\mathcal{D}. \label{eq:C_zero_main}}
To show this, we extend the real variables $k_x$ and $k_y$
to independent complex variables $K_x,K_y\in\mathbb C$.
Define $Y=K_x-\ii K_y$ and $W=K_x+\ii K_y$, and then the identity theorem implies that \cref{eq:C_sum_main} can be extended to
\eq{\sum_{\boldsymbol\rho\in\mathcal D}
e^{-\ii  Y (\rho_x + \ii \rho_y)/2 } e^{-\ii  W (\rho_x - \ii \rho_y)/2 }
C_{\alpha,\boldsymbol\rho}(Y)=0
\label{eq:C_sum_complexified_main}} in a small complex neighborhood that contains $I_x\times I_y$.
Since $e^{-\ii  W (\rho_x - \ii \rho_y)/2 }$ is linearly independent for finitely many different $\bsl{\rho}$, $e^{-\ii Y (\rho_x + \ii \rho_y)/2}
C_{\alpha,\bsl{\rho}}(Y)
= 0$ holds separately for each $\bsl{\rho}$, which leads to \cref{eq:C_zero_main}, as shown in \cref{subsec:Isolated Chern Ideal Bands}. Note that this argument is not necessarily applicable to exponentially decaying hopping as it has infinitely many distinct $\bsl{\rho}$. Using $w_a(k)=1$, \cref{eq:C_zero_main} is equivalent to
\eq{A_{\boldsymbol\rho} w(k)
 =\lambda_{\boldsymbol\rho}(k) w(k),
 \qquad
 \lambda_{\boldsymbol\rho}(k)\equiv[A_{\boldsymbol\rho} w(k)]_a.
 \label{eq:Arho_w_main}}
Notice that for every fixed $\boldsymbol\rho$, $\lambda_{\boldsymbol\rho}(k)$ is holomorphic in the neighborhood $I_x\times I_y$ and is an eigenvalue of the fixed finite matrix $A_{\boldsymbol\rho}$. 
It therefore takes values in a finite set and must be constant on the neighborhood $I_x\times I_y$, leading to
\eq{A_{\boldsymbol\rho} w(k)=\lambda_{\boldsymbol\rho} w(k). \label{eq:lambda_constant_main}}
Then, if $w(k)$ is nonconstant, different values of $w(k)$ in the neighborhood must form a basis of dimension at least two, since $w_a(k)=1$ is fixed.
Define $v_1=w(k_1)$ and $v_2=w(k_2)$ to be two independent vectors that satisfy \cref{eq:lambda_constant_main}.
Then
\eq{
h(\bsl{k})v_{i}
=
\sum_{\bsl{\rho}\in\mathcal D}
e^{-\ii\bsl{k}\cdot\bsl{\rho}}
\lambda_{\bsl{\rho}}v_{i} =
E_{\bsl{k}}v_{i} \ ,
}
for both $i=1,2$ and for all $\bsl{k}$ in the neighborhood $I_x \times I_y$.
This degeneracy contradicts the assumption that $E_{\bsl{k}}$ is isolated, and thus $w(k)$ must also be constant in the neighborhood.
But a constant $w(k)$ would give a constant $P_{\bsl{k}}$ which leads to $\Tr[g_{\bsl{k}_0}]=0$, contradicting $\Tr[g_{\bsl{k}_0}]>0$.
Therefore, there cannot be Chern-ideal isolated bands in a finite-range-hopping tight-binding model, regardless of the number of orbitals and the embedding.

Since our proof is local in $\bsl{k}$, it also applies to the critical
bands proposed in
\cite{li2026stabletopologyexactlyflat,
liu2026theorytopologicaltopologicalflatbands},
where the target band touches other bands at finitely many isolated
momenta while its projector remains continuous and the Berry curvature remains finite across the touching points.
For such Chern-ideal critical bands with positive Chern number, there must exist a $\bsl{k}_0$ away from these touching points where $\Tr[g_{\bsl{k}_0}]=F_{\bsl{k}_0}>0$.  Following the same arguments in the preceding proof, this leads to band degeneracy away from the assumed touching points or $\Tr[g_{\bsl{k}_0}]=0$, resulting in a contradiction.
Therefore, critical bands with isolated touching points and non-diverging Berry curvature cannot be Chern-ideal either.

\section{WL-Ideal Bands with Zero Total Chern number}

Now we discuss the WL-ideal bands~\cite{Yu_2025_WL_ideal}.
We focus on a generic isolated set of two WL-ideal bands with zero total Chern numbers, nontrivial normal WL winding and non-singular non-Abelian Berry curvature.
As shown in \refcite{Yu_2025_WL_ideal}, such WL-ideal bands can always be unitarily transformed into two Chern-ideal states with opposite Chern numbers.
The Chern-ideal states have the same embedding as the parent WL-ideal band, since they are eigenstates of the non-abelian Berry curvature.
Therefore, the considered WL-ideal bands cannot exist in any finite-band tight-binding model with finite-range hopping.
The construction of the considered WL-ideal bands in finite-band tight-binding models with exponentially decaying hopping depends on the symmetries considered.
For the Kane-Mele $\dsZ_2$-ideal bands protected by the spinful time-reversal symmetry, the construction can be naturally done by combining a two-band model with a $\Ch=1$ Chern-ideal band with its time-reversal partner.
Comprehensive constructions covering all symmetry cases are left for future work.

\section{Conclusion and Discussion}\label{sec:Conclusion and Discussion}

We provide a construction of $|\Ch|=1$ Chern-ideal bands in finite-band tight-binding model with exponentially decaying hopping, and show that no Chern-ideal bands can exist when the hopping becomes finite range.
The construction and obstruction are generalized to some WL-ideal bands with zero-total Chern numbers.
Our results provide helpful guidance for future constructions of analytic tight-binding models for fractional Chern insulators.

\section{Acknowledgments }
We thank Jonah Herzog-Arbeitman, Steven Simon, Zhida Song, Ziwei Wang, and Fengcheng Wu for helpful discussions.
This work was supported by J.Y.’s startup funds from the University of Florida.
S.R. is supported by a Simons
Investigator Grant from the Simons Foundation (Award
No. 566116). This work is supported by the Gordon
and Betty Moore Foundation EPiQS initiative, Grant
GBMF8685.01.

\emph{Note added:} During the final stage of our work, we became aware of \refcite{Li_2026_Lattice_GLL,Other_2026_to_appear}, which overlap with our work. The overlapping results of our work and \refcite{Li_2026_Lattice_GLL,Other_2026_to_appear} agree.

\bibliography{bibfile_references_canonical,bib2.bib}

\appendix
\onecolumngrid
\clearpage

\tableofcontents

\section{Setup}
\label{app: basics}

We consider a two-dimensional Bravais lattice $\mathcal{L}$, with
lattice vectors denoted by $\bsl{R}\in\mathcal{L}$.
In each unit cell, $\alpha=1,\ldots,N$ labels the internal degrees of
freedom, which we collectively refer to as ``orbitals.''
Here, we use the term ``orbital'' in a broad sense: besides atomic
orbitals, $\alpha$ may also label other internal degrees of freedom
such as spin, valley, or sublattice.
The single-particle Hilbert space is spanned by the orthonormal basis
states $\ket{\bsl{R},\alpha}$ satisfying $\braket{\bsl{R},\alpha}{\bsl{R}',\alpha'}
= \delta_{\bsl{R}\bsl{R}'}
\delta_{\alpha\alpha'}$.
We use $\bsl{\tau}_{\alpha}$ to denote the physical position of orbital
$\alpha$ relative to the origin of the $\bsl{R}=0$ unit cell.
Therefore, the physical position of orbital $\alpha$ in unit cell
$\bsl{R}$ is $\bsl{R}+\bsl{\tau}_{\alpha}$. The general tight-binding Hamiltonian takes the
form
\eq{
H = \sum_{\bsl{R}\bsl{R}'}\sum_{\alpha\alpha'} \left[t_{\bsl{R}-\bsl{R}'}\right]_{\alpha\alpha'}\ket{\bsl{R},\alpha}\bra{\bsl{R}',\alpha'}\ ,
}
where $t_{\bsl{R}-\bsl{R}'}$ only depends on $\bsl{R}-\bsl{R}'$ due to the lattice translation symmetry.
Since there are $N$ orbitals in each unit cell, we refer to this as an
$N$-orbital, or equivalently an $N$-band, model.

To transform to momentum space, we need to define the form of the Fourier transformation, where the concept of ``embedding'' enters.
Specifically, we define
\eq{
\ket{\bsl{k},\alpha} = \frac{1}{\sqrt{N_k}}\sum_{\bsl{R}}e^{\ii(\bsl{R}+\bsl{\xi}_{\alpha})\cdot\bsl{k}} \ket{\bsl{R},\alpha}\ ,
\label{eq:embedded_fourier_transform}
}
where $N_k$ is the number of unit cells (which is equal to the number of $\bsl{k}$ points), we call $\bsl{\xi}_{\alpha}$ the embedded position of the orbital $\alpha$ in $\bsl{R}=0$ unit cell.
The embedded positions specify
the transformation of the orbital Bloch basis under reciprocal-lattice
translations. For any reciprocal-lattice vector $\bsl{G}$,
\eq{
\ket{\bsl{k}+\bsl{G},\alpha}
=
e^{\ii\bsl{G}\cdot\bsl{\xi}_{\alpha}}
\ket{\bsl{k},\alpha},
\label{eq:embedded_basis_covariance}
}
where we used $e^{\ii\bsl{G}\cdot\bsl{R}}=1$ for every
$\bsl{R}\in\mathcal{L}$.

Therefore, the choice $\bsl{\xi}_{\alpha}=0$ for all $\alpha$ gives the periodic
embedding, for which the orbital Bloch basis is periodic under
reciprocal-lattice translations.
The choice
$\bsl{\xi}_{\alpha}=\bsl{\tau}_{\alpha}$ gives the physical embedding,
for which the embedded positions are the same with the physical orbital
positions \cite{haldane2014attachmentsurfacefermiarcs,Lim_2015,Simon_2020,Fuchs_2021,Huhtinen_2022,PhysRevLett.133.246603}.
In the following, we keep the embedded positions
$\{\bsl{\xi}_{\alpha}\}$ arbitrary and derive all results using
\cref{eq:embedded_fourier_transform}. The periodic and physical
embeddings are therefore both included automatically.

The Hamiltonian in momentum space reads
\eq{
H = \sum_{\bsl{k}\in\BZ}\sum_{\alpha,\alpha'} \ket{\bsl{k},\alpha} \left[h(\bsl{k})\right]_{\alpha\alpha'}\bra{\bsl{k},\alpha'}\ , 
\label{eq:Hamiltonian_main_tb}
}
where
\eqa{
\left[h(\bsl{k})\right]_{\alpha\alpha'}
&=
\bra{\bsl{k},\alpha}
H
\ket{\bsl{k},\alpha'}
\\
&=
\left[
\frac{1}{\sqrt{N_k}}
\sum_{\bsl{R}\in\mathcal L}
e^{-\ii(\bsl{R}+\bsl{\xi}_{\alpha})\cdot\bsl{k}}
\bra{\bsl{R},\alpha}
\right]
H
\left[
\frac{1}{\sqrt{N_k}}
\sum_{\bsl{R}'\in\mathcal L}
e^{\ii(\bsl{R}'+\bsl{\xi}_{\alpha'})\cdot\bsl{k}}
\ket{\bsl{R}',\alpha'}
\right]
\\
&=
\frac{1}{N_k}
\sum_{\bsl{R},\bsl{R}'\in\mathcal L}
e^{-\ii(\bsl{R}+\bsl{\xi}_{\alpha})\cdot\bsl{k}}
e^{\ii(\bsl{R}'+\bsl{\xi}_{\alpha'})\cdot\bsl{k}}
\bra{\bsl{R},\alpha}
H
\ket{\bsl{R}',\alpha'}
\\
&=
\frac{1}{N_k}
\sum_{\bsl{R},\bsl{R}'\in\mathcal L}
e^{-\ii(\bsl{R}+\bsl{\xi}_{\alpha})\cdot\bsl{k}}
e^{\ii(\bsl{R}'+\bsl{\xi}_{\alpha'})\cdot\bsl{k}}
\left[
t_{\bsl{R}-\bsl{R}'}
\right]_{\alpha\alpha'}
\\
&=
\sum_{\bsl{\Delta R}\in\mathcal L}
\left[
t_{\bsl{\Delta R}}
\right]_{\alpha\alpha'}
e^{-\ii\bsl{k}\cdot
\left(
\bsl{\Delta R}
+\bsl{\xi}_{\alpha}
-\bsl{\xi}_{\alpha'}
\right)} .
\label{eq:h_from_t_general_embedding}
}
We let $
\bsl{\Delta R}
\equiv
\bsl{R}-\bsl{R}'$ and in the last equality we used
$\sum_{\bsl{R}'\in\mathcal L}1=N_k$.
Conversely, the real-space hopping matrices are recovered from the
inverse Fourier transformation,
\eq{
\left[t_{\bsl{\Delta R}}\right]_{\alpha\alpha'}
=
\frac{1}{N_k}
\sum_{\bsl{k}\in\BZ}
e^{\ii\bsl{k}\cdot
\left(
\bsl{\Delta R}
+\bsl{\xi}_{\alpha}
-\bsl{\xi}_{\alpha'}
\right)}
\left[h(\bsl{k})\right]_{\alpha\alpha'} .
\label{eq:t_from_h_general_embedding}
}
The Bloch momenta may be represented by
$\bsl{k}\in\BZ$, while
\cref{eq:h_from_t_general_embedding} defines $h(\bsl{k})$ for every
$\bsl{k}\in\dsR^2$. For any reciprocal-lattice vector $\bsl{G}$, using
$e^{-\ii\bsl{G}\cdot\bsl{\Delta R}}=1$ for every
$\bsl{\Delta R}\in\mathcal L$, we obtain
\eqa{
\left[h(\bsl{k}+\bsl{G})\right]_{\alpha\alpha'}
&=
\sum_{\bsl{\Delta R}\in\mathcal L}
\left[t_{\bsl{\Delta R}}\right]_{\alpha\alpha'}
e^{-\ii(\bsl{k}+\bsl{G})\cdot
\left(
\bsl{\Delta R}
+\bsl{\xi}_{\alpha}
-\bsl{\xi}_{\alpha'}
\right)}
\\
&=
e^{-\ii\bsl{G}\cdot
\left(
\bsl{\xi}_{\alpha}
-\bsl{\xi}_{\alpha'}
\right)}
\sum_{\bsl{\Delta R}\in\mathcal L}
\left[t_{\bsl{\Delta R}}\right]_{\alpha\alpha'}
e^{-\ii\bsl{k}\cdot
\left(
\bsl{\Delta R}
+\bsl{\xi}_{\alpha}
-\bsl{\xi}_{\alpha'}
\right)}
\\
&=
e^{-\ii\bsl{G}\cdot
\left(
\bsl{\xi}_{\alpha}
-\bsl{\xi}_{\alpha'}
\right)}
\left[h(\bsl{k})\right]_{\alpha\alpha'}.
\label{eq:generalEmbedding}
}
Equivalently,
\eq{
h(\bsl{k}+\bsl{G})
=
V_{\bsl{G}}
h(\bsl{k})
V_{\bsl{G}}^\dagger,
\label{eq:h_embedding}
}
where
\eq{
V_{\bsl{G}}
\equiv
\operatorname{diag}
\left(
e^{-\ii\bsl{G}\cdot\bsl{\xi}_{1}},
e^{-\ii\bsl{G}\cdot\bsl{\xi}_{2}},
\ldots,
e^{-\ii\bsl{G}\cdot\bsl{\xi}_{N}}
\right).
\label{eq:general_embedding_matrix}
}

\subsection{Analyticity, Locality, and Isolated-Band Projectors}

\begin{theorem}
\label{thm:exp_decay_analyticity}
For a finite-band tight-binding model, the real-space hopping
matrices $t_{\bsl{R}}$ decay exponentially if and only if $h(\bsl{k})$ is real-analytic in $\bsl{k}$ .
\end{theorem}

\begin{proof}
It was proven in \refcite{Kuchment_2008,Kuchment_2016}, and we rewrite the proof here.
Define $D_{\bsl{\xi}}(\bsl{k})
=
\operatorname{diag}
\left(
e^{-\ii\bsl{k}\cdot\bsl{\xi}_1},
\ldots,
e^{-\ii\bsl{k}\cdot\bsl{\xi}_N}
\right)$ and let
\eq{
\widetilde h(\bsl{k}) \equiv
D_{\bsl{\xi}}^{-1}(\bsl{k})
h(\bsl{k})
D_{\bsl{\xi}}(\bsl{k})
=
D_{\bsl{\xi}}^\dagger(\bsl{k})
h(\bsl{k})
D_{\bsl{\xi}}(\bsl{k}),
}
where the second equality holds for real $\bsl{k}$.
Using \cref{eq:h_from_t_general_embedding}, we obtain
\eqa{
\left[\widetilde h(\bsl{k})\right]_{\alpha\alpha'}
&=
\left[
D_{\bsl{\xi}}^{-1}(\bsl{k})
h(\bsl{k})
D_{\bsl{\xi}}(\bsl{k})
\right]_{\alpha\alpha'} =
e^{\ii\bsl{k}\cdot\bsl{\xi}_{\alpha}}
\left[h(\bsl{k})\right]_{\alpha\alpha'}
e^{-\ii\bsl{k}\cdot\bsl{\xi}_{\alpha'}}
\\
&=
e^{\ii\bsl{k}\cdot\bsl{\xi}_{\alpha}}
\sum_{\bsl{R}\in\mathcal L}
\left[t_{\bsl{R}}\right]_{\alpha\alpha'}
e^{-\ii\bsl{k}\cdot
\left(
\bsl{R}
+\bsl{\xi}_{\alpha}
-\bsl{\xi}_{\alpha'}
\right)}
e^{-\ii\bsl{k}\cdot\bsl{\xi}_{\alpha'}}
\\
&=
\sum_{\bsl{R}\in\mathcal L}
\left[t_{\bsl{R}}\right]_{\alpha\alpha'}
e^{-\ii\bsl{k}\cdot\bsl{R}}.
}
Here we let $\bsl R$ denote the previous $\bsl{\Delta R}$ index. Since $D_{\bsl{\xi}}(\bsl{k})$ is analytic and invertible,
$h(\bsl{k})$ and $\widetilde h(\bsl{k})$ have the same local
regularity.

Suppose first that the hopping matrices decay exponentially, namely,
that there exist constants $C,\mu>0$ such that
\eq{
\norm{t_{\bsl{R}}}
\leq
C e^{-\mu\lvert\bsl{R}\rvert}
\qquad
\text{for every }\bsl{R}\in\mathcal{L},
\label{eq:exponential_hopping_bound}
}
where $\norm{t_{\boldsymbol R}}$ denotes the largest absolute value of
$\left[t_{\boldsymbol R}\right]_{\alpha\alpha'}$ over all orbital indices
$\alpha$ and $\alpha'$. 
For any nonnegative integers
$m_x,m_y$, we have
\eq{
\sum_{\boldsymbol{R}\in\mathcal L}
\norm{
(-\ii R_x)^{m_x}
(-\ii R_y)^{m_y}
t_{\boldsymbol{R}}
e^{-\ii\boldsymbol{k}\cdot\boldsymbol{R}}
}
\leq
C
\sum_{\boldsymbol{R}\in\mathcal L}
\lvert\boldsymbol{R}\rvert^{m_x+m_y}
e^{-\mu\lvert\boldsymbol{R}\rvert}
<\infty.
}
where we used $\lvert R_x\rvert,\lvert R_y\rvert\leq\lvert\bsl R\rvert$, $\lvert e^{-\ii\bsl k\cdot\bsl R}\rvert=1$. The bound is independent of $\boldsymbol{k}$, so the series obtained
after differentiating also converges absolutely and uniformly by the Weierstrass $M$-test. Therefore, the differentiation is
valid, and
\eqa{
\partial_{k_x}^{m_x}
\partial_{k_y}^{m_y}
\widetilde h(\boldsymbol{k})
=
\sum_{\boldsymbol{R}\in\mathcal L}
(-\ii R_x)^{m_x}
(-\ii R_y)^{m_y}
t_{\boldsymbol{R}}
e^{-\ii\boldsymbol{k}\cdot\boldsymbol{R}}.
\label{eq:derivative_series}
}
Since this holds for every $m_x,m_y$,
$\widetilde h(\boldsymbol{k})$ is smooth.
As $D_{\bsl{\xi}}(\bsl{k})$ is smooth in $\bsl{k}$, we have $h(\bsl{k})$ is smooth.

More strongly, consider a complex momentum
$\bsl{k}+\ii\bsl{\kappa}$, where
$\bsl{\kappa}\in\dsR^2$. Then
\eq{
\norm{
t_{\bsl{R}}
e^{-\ii(\bsl{k}+\ii\bsl{\kappa})\cdot\bsl{R}}
} =
\norm{t_{\bsl{R}}}
e^{\bsl{\kappa}\cdot\bsl{R}} \leq
C e^{-\mu\lvert\bsl{R}\rvert}
e^{\lvert\bsl{\kappa}\rvert\lvert\bsl{R}\rvert} =
C e^{-(\mu-\lvert\bsl{\kappa}\rvert)
\lvert\bsl{R}\rvert},
}
where we used
$\bsl{\kappa}\cdot\bsl{R}
\leq
\lvert\bsl{\kappa}\rvert\lvert\bsl{R}\rvert$.
For any $0<\mu'<\mu$ and
$\lvert\bsl{\kappa}\rvert\leq\mu'$, we have
\eq{
C e^{-(\mu-\lvert\bsl{\kappa}\rvert)
\lvert\bsl{R}\rvert}
\leq
C e^{-(\mu-\mu')\lvert\bsl{R}\rvert}.
}
The right-hand side is independent of $\bsl{k}$ and
$\bsl{\kappa}$, and its sum over
$\bsl{R}\in\mathcal L$ converges.
Therefore, by the Weierstrass $M$-test, the complexified Fourier
series converges uniformly for
$\lvert\bsl{\kappa}\rvert\leq\mu'$.
Since each term is holomorphic in the complex momentum variables, the
sum is also holomorphic. Hence $\widetilde h(\bsl{k})$ admits a
holomorphic extension to a neighborhood of real momentum space and is
therefore real-analytic.

Conversely, suppose that $ h(\bsl{k})$ is real-analytic on
$\mathbb{R}^2$, which implies that $\widetilde h(\bsl{k})$ is real-analytic on
$\mathbb{R}^2$.
Then, $\widetilde h(\bsl{k})$ admits a holomorphic
extension $\widetilde h(\bsl{k}+\ii\bsl{\kappa})$ to a complex neighborhood of the BZ, then we can choose $\lvert\bsl{\kappa}\rvert<\mu$ for some fixed $\mu>0$ over the $\BZ$, as  $\BZ$ is a compact space.
In the thermodynamic limit,
\eq{
t_{\bsl{R}}
=
\frac{1}{\vol(\BZ)}
\int_{\BZ}
\widetilde h(\bsl{k})
e^{\ii\bsl{k}\cdot\bsl{R}}
\dd^2k.
\label{eq:inverse_fourier_h_t}
}
For $\bsl{R}\neq0$, choose any $0<\mu'<\mu$ and define $\bsl{\kappa}_{\bsl{R}}
= \mu' \frac{\bsl{R}}{\lvert\bsl{R}\rvert}.$
Then $\bsl{\kappa}_{\bsl{R}}\cdot\bsl{R} = \mu'\lvert\bsl{R}\rvert.$
Because $\widetilde h$ is holomorphic and reciprocal-lattice
periodic, the integration contour in
\cref{eq:inverse_fourier_h_t} may be shifted from $\bsl{k}$ to
$\bsl{k}+\ii\bsl{\kappa}_{\bsl{R}}$. Hence
\eq{
t_{\bsl{R}} =
\frac{1}{\vol(\BZ)}
\int_{\BZ}
\widetilde h
\left(
\bsl{k}+\ii\bsl{\kappa}_{\bsl{R}}
\right)
e^{\ii
\left(
\bsl{k}+\ii\bsl{\kappa}_{\bsl{R}}
\right)\cdot\bsl{R}}\dd^2k =
e^{-\mu'\lvert\bsl{R}\rvert}
\frac{1}{\vol(\BZ)}
\int_{\BZ}
\widetilde h
\left(
\bsl{k}+\ii\bsl{\kappa}_{\bsl{R}}
\right)
e^{\ii\bsl{k}\cdot\bsl{R}} \dd^2k.
}
Since $\widetilde h$ is continuous and reciprocal-lattice periodic,
it is bounded on $\bsl{k}\in\BZ, \lvert\bsl{\kappa}\rvert\leq\mu'.$
Thus, there exists a constant $C_{\mu'}>0$ such that $\norm{
\widetilde h(\bsl{k}+\ii\bsl{\kappa})
} \leq
C_{\mu'}$. Using
$\lvert e^{\ii\bsl{k}\cdot\bsl{R}}\rvert=1$, we obtain
\eq{
\norm{t_{\bsl{R}}} \leq
e^{-\mu'\lvert\bsl{R}\rvert}
\frac{1}{\vol(\BZ)}
\int_{\BZ}
\norm{
\widetilde h
\left(
\bsl{k}+\ii\bsl{\kappa}_{\bsl{R}}
\right)
}
\dd^2k \leq
C_{\mu'}
e^{-\mu'\lvert\bsl{R}\rvert}.
}
The coefficient is finite and can be absorbed into the
constant. Therefore, the hopping matrices decay exponentially. This proves both directions of the equivalence.
\end{proof}

We next consider the projector of an energy band. Suppose $E_{\bsl{k}}$ is an energy band of interest, and $U_{\bsl{k}}$ is the corresponding eigenvector:
\eq{
h(\bsl{k}) U_{\bsl{k}} = E_{\bsl{k}} U_{\bsl{k}}\ .
\label{eq:eigeneq}
}
We define $E_{\bsl{k}}$ to be isolated at $\bsl{k}$ if $E_{\bsl{k}}$ is different from all other eigenvalues of $h(\bsl{k})$.
If $E_{\bsl{k}}$ of $h(\bsl{k})$ is isolated at all $\bsl{k}$, we call $E_{\bsl{k}}$ an isolated band.
Owing to \cref{eq:h_embedding}, we only consider bands with $E_{\bsl{k}+\bsl{G}} = E_{\bsl{k}}$ and $\left| U_{\bsl{k}+\bsl{G}}^\dagger  V_{\bsl{G}}U_{\bsl{k}} \right| = 1$. An isolated band always satisfies the condition.

\begin{theorem}
\label{thm:P_smooth_at_Isolated}
Let $h(\bsl{k})$ be real-analytic in an open neighborhood of
$\bsl{k}_0\in\dsR^2$. 
If $E_{\bsl{k}}$ is isolated at $\bsl{k}_0\in \dsR^2 $, then $P_{\bsl{k}}=U_{\bsl{k}} U_{\bsl{k}}^\dagger$ is real-analytic in an open neighborhood
of $\bsl{k}_0$
\end{theorem}
\begin{proof}
It was proven in \refcite{Panati_2013,Monaco_2018}, and we rewrite it here. First, we express the projection matrix in term of the resolvent.
Let us define a contour
\eq{
\gamma = \{ z \in \dsC \mid |z-E_{\bsl{k}_0}| = \epsilon \}
}
with small positive $\epsilon$.
Since $E_{\bsl{k}_0}$ is isolated, we choose $\epsilon>0$ to be
less than half of the minimum distance between $E_{\bsl{k}_0}$ and the other eigenvalues at $\bsl k_0$.
Since $E_{\bsl{k}}$ is continuous in $\bsl{k}$, there exists a sufficiently small
open neighborhood $\mathcal U$ of $\bsl{k}_0$ such that, for every
$\bsl{k}\in\mathcal U$, exactly one eigenvalue of $h(\bsl{k})$ lies
inside $\gamma$, while no eigenvalue lies on $\gamma$. The corresponding spectral projector is therefore given by the Riesz
formula \cite{Monaco_2018},
\eq{
\label{eq:Riesz}
P_{\bsl{k}}
=
\frac{1}{2\pi i}
\oint_{\gamma}
[z-h(\bsl{k})]^{-1}dz
}
for all $\bsl{k}$ in the sufficiently small neighborhood $\mathcal U$ of $\bsl{k}_0$.

To verify that, we note first that $z-h(\bsl{k})$ is invertible for $z\in \gamma$ and all $\bsl{k}\in \mathcal U$ owing to the choice of $\epsilon$, and thus the integral is finite for all $\bsl{k}$ in the infinitesimal neighborhood of $\bsl{k}_0$.
Then, for all $\bsl{k}$ in the infinitesimal neighborhood $\mathcal U$ of $\bsl{k}_0$, we have
\eq{
\left[ \frac{1}{2\pi i}
\oint_{\gamma}
(z-h(\bsl{k}))^{-1}dz  \right] U_{\bsl{k}} = \left[ \frac{1}{2\pi i}
\oint_{\gamma}
(z-E_{\bsl{k}})^{-1}dz  \right] U_{\bsl{k}} = U_{\bsl{k}} 
}
and
\eq{
\left[ \frac{1}{2\pi i}
\oint_{\gamma}
(z-h(\bsl{k}))^{-1}dz  \right] V_{\bsl{k}} = \left[ \frac{1}{2\pi i}
\oint_{\gamma}
(z-E_{\bsl{k}}')^{-1}dz  \right] V_{\bsl{k}} = 0 \ ,}
where $V_{\bsl{k}}$ is any eigenvector of $h(\bsl{k})$ that is orthogonal to $U_{\bsl{k}}$, and the equation holds because $V_{\bsl{k}}$ has eigenvalue $E'_{\bsl{k}}$ that lies outside $\gamma$.
This verifies
\cref{eq:Riesz}.

It remains to determine the regularity of the projector. For every
fixed $z\in\gamma$, the matrix $A_z(\bsl{k}) \equiv z-h(\bsl{k})$
is real-analytic in $\bsl{k}$ and invertible throughout
$\mathcal U$. Its inverse is given by the adjugate formula
\eq{
A_z(\bsl{k})^{-1}
=
\frac{
\operatorname{adj}A_z(\bsl{k})
}{
\det A_z(\bsl{k})
}
=
\frac{
\operatorname{adj}\left[z-h(\bsl{k})\right]
}{
\det\left[z-h(\bsl{k})\right]
}.
\label{eq:real_analytic_inverse_adjugate}
}
The entries of
$\operatorname{adj}A_z(\bsl{k})$ and
$\det A_z(\bsl{k})$ are polynomial expressions in the entries of
$A_z(\bsl{k})$, and are therefore real-analytic in $\bsl{k}$.
Moreover, since $A_z(\bsl{k})$ is invertible,
$\det A_z(\bsl{k})\neq0$ throughout $\mathcal U$. Hence
$1/\det A_z(\bsl{k})$ is also real-analytic, and therefore $\left[z-h(\bsl{k})\right]^{-1}$
is real-analytic in $\bsl{k}$. 
Finally, $\gamma$ is a fixed compact contour, so the contour integral
in \cref{eq:Riesz} preserves real analyticity. It follows that
$P_{\bsl{k}}$ is real-analytic in a neighborhood of
$\bsl{k}_0$.
\end{proof}

By \cref{thm:P_smooth_at_Isolated}, if $h(\bsl{k})$ is
real-analytic on $\dsR^2$, then the projector of an isolated band is
real-analytic on $\dsR^2$.

\begin{theorem}
\label{thm:analytic_projector_parent_hamiltonian}
$P_{\bsl k}$ is real-analytic if and only if it is the isolated-single-band projector of some $h(\bsl{k})$ with exponentially decaying hopping.
\end{theorem}
\begin{proof}

Combining the preceding results from
\cref{thm:exp_decay_analyticity} 
\cite{Kuchment_2008,Kuchment_2016}
and
\cref{thm:P_smooth_at_Isolated} 
\cite{Panati_2013,Monaco_2018}
gives this useful conclusion. We include the short proof for
completeness. Suppose first that $P_\bsl{k}$ is a real-analytic rank-one projector satisfying $P_{\bsl{k}+\bsl{G}} = V_{\bsl{G}}P_\bsl{k}V_{\bsl{G}}^\dagger.$ Define the parent Bloch Hamiltonian
\eq{
h_P(\bsl{k}) = \mathbb{1}_N-2P_{\bsl k}.
\label{eq:flattened_parent_hamiltonian}
}
Since $P^2=P$, 
\eq{h_P(\bsl{k})P_{\bsl k} =
-P_{\bsl k}, \quad h_P(\bsl{k}) \left[\mathbb{1}_N-P_{\bsl k}\right]
= \mathbb{1}_N-P_{\bsl k}.}
So $P_{\bsl k}$ is the spectral projector of the eigenvalue $-1$. And because of \cref{eq:flattened_parent_hamiltonian}, $h_P(\bsl{k})$ has the same regularity with $P_{\bsl k}$ and it satisfies
\eq{
h_P(\bsl{k}+\bsl{G})
=V_{\bsl{G}}h_P(\bsl{k})V_{\bsl{G}}^\dagger.
}
The corresponding real-space hopping matrices are
\eq{
\left[t^{P}_{\bsl{R}}\right]_{\alpha\alpha'}
= \frac{1}{\vol(\BZ)}
\int_{\BZ}
\dd^2k
e^{\ii\bsl{k}\cdot
\left(
\bsl{R}
+\bsl{\xi}_{\alpha}
-\bsl{\xi}_{\alpha'}
\right)}
\left[h_P(\bsl{k})\right]_{\alpha\alpha'}.
}
Because of \cref{eq:flattened_parent_hamiltonian}, $h_P(\bsl{k})$ is
real-analytic whenever $P_{\bsl{k}}$ is real-analytic.
By the analyticity--exponential-decay correspondence in \cref{thm:exp_decay_analyticity}, the
real-space hopping amplitudes of $h_P(\bsl{k})$ decay exponentially.
Thus, $h_P(\bsl{k})$ is an exponentially decaying parent Hamiltonian
whose isolated $-1$ band has projector $P_{\bsl{k}}$.

Conversely, suppose that an exponentially decaying tight-binding
Hamiltonian  has an isolated rank-one band. By
\cref{thm:exp_decay_analyticity}, its Bloch Hamiltonian
$h(\bsl{k})$ is real-analytic. The analyticity of the corresponding
isolated-band projector then follows from
\cref{thm:P_smooth_at_Isolated}.
\end{proof}

\subsection{Quantum Geometry and Ideality}
\label{app:quantum-geometry}

We now define the quantum geometry of a band $E_{\bsl{k}}$ with projector $P_{\bsl{k}}$. For $i,j\in\{x,y\}$, the quantum geometric tensor is defined locally by \cite{Provost_1980,Kolodrubetz_2017}
\eq{
\mathcal{Q}_{ij}(\bsl{k})
=
\left[
\partial_{k_i}U_{\bsl k}
\right]^\dagger
\left[
\mathbb{1}_N-P_{\bsl k}
\right]
\left[
\partial_{k_j}U_{\bsl k}
\right].
\label{eq:quantum_geometric_tensor_vector}
}
This expression is invariant under the local band-gauge transformation $U_{\bsl k}
\longmapsto
U_{\bsl k} e^{\ii\chi(\bsl{k})}.$
Its symmetric part defines the quantum metric~\cite{Fubini1904,Study1905},
\eq{
\left[g_{\bsl{k}}\right]_{ij}
\equiv
\frac{1}{2}\left[ \mathcal{Q}_{ij}(\bsl{k}) + \mathcal{Q}_{ji}(\bsl{k})\right]
=
\frac{1}{2}
\Tr\left[
\partial_{k_i}P_{\bsl k}
\partial_{k_j}P_{\bsl k}
\right],
\label{eq:quantum_metric_projector}
}
while its antisymmetric part defines the Berry curvature~\cite{TKNN},
\eq{
F_{ij}(\bsl{k})
\equiv
\ii \left[ \mathcal{Q}_{ij}(\bsl{k}) - \mathcal{Q}_{ji}(\bsl{k})\right]
=
\ii\Tr\left[
P_{\bsl k}
\left[
\partial_{k_i}P_{\bsl k},
\partial_{k_j}P_{\bsl k}
\right]
\right].
\label{eq:berry_curvature_projector}
}
In two dimensions, we write $F_{\bsl{k}}
\equiv
F_{xy}(\bsl{k}).$
A pointwise inequality holds
\eq{
\Tr[g_{\bsl{k}}] \geq
\left|F_{\bsl{k}}\right|,
\label{eq:trace_inequality}
}
where $\Tr[g_{\bsl{k}}] \equiv g_{xx}(\bsl{k})+g_{yy}(\bsl{k})$. 
The first Chern number of the band is
\eq{
\Ch
=
\frac{1}{2\pi}
\int_{\BZ}
F_{\bsl{k}}
\dd^2k
\in\dsZ.
\label{eq:chern_number_definition}
}
We call the band Chern-ideal if it saturates this inequality at every
momentum ~\cite{Jie2021IdealBands,Ledwith_2022,Wang_2022,Parker2023IdealBands,Valentin2023IdealBands,Wang_2023,Dong_2023_ideal_Higher_Chern}:
\eq{
\frac{1}{2\pi}\int_{\BZ} \Tr[g_{\bsl{k}}] d^2 k
= |\Ch|\ .
\label{eq:ideal_band_definition}
}

Following the geometric interpretation in
\refcite{Yu_2025_WL_ideal}, we regard
$\bsl{k}\mapsto P_{\bsl{k}}$ as a smooth map from the BZ to the Grassmannian
$\mathrm{Gr}(1,\mathbb{C}^{M})$ of rank-1 projectors.  When $\mathrm{Gr}(1,\mathbb{C}^{M})$ is equipped with the metric inherited from the
ambient space of Hermitian matrices,
the metric induced on momentum
space is precisely the quantum metric in \cref{eq:quantum_metric_projector}, then in flat Cartesian momentum coordinates, the corresponding Dirichlet functional $E_{\mathrm D}[P]$ becomes $\int_{\mathrm{BZ}} d^2k\Tr[g_{\bsl{k}}]
= \Tr\mathcal{G}.$ Therefore, $\Tr\mathcal{G}$ is exactly the Dirichlet functional of the projector map over BZ.
Interested readers may refer to \refcite{Yu_2025_WL_ideal} for more detailed derivations.

In this work, we restrict to 
\eq{
\Ch>0 \Leftrightarrow
\Tr[g_{\boldsymbol{k}}]=F_{\boldsymbol{k}}
\quad
\text{for every }\boldsymbol{k}\in\BZ,
\label{eq:chern_ideal_band_definition}
}
when studying Chern-ideal bands, since the case $\Ch<0$ is obtained by time-reversal of the Bloch states. Although the embedding is a choice of Bloch basis, the
pointwise quantum geometry generally depends on this choice. For a
change of embedded positions
$\boldsymbol{\xi}_{\alpha}\to\boldsymbol{\xi}'_{\alpha}$, define
\eq{
W_{\bsl k}
=
\operatorname{diag}
\left(
e^{-\ii\boldsymbol{k}\cdot
(\boldsymbol{\xi}'_1-\boldsymbol{\xi}_1)},
\ldots,
e^{-\ii\boldsymbol{k}\cdot
(\boldsymbol{\xi}'_N-\boldsymbol{\xi}_N)}
\right).
}
The orbital coefficient vector and projector transform as
\eq{
U'_{\bsl{k}} = W_{\bsl k}U_{\bsl k},
\qquad
P'_{\bsl k} = W_{\bsl k}P_{\bsl k}W^\dagger_{\bsl k}.
\label{eq:projector_embedding_change}
}
Because $W_{\bsl k}$ depends on momentum, $\partial_{k_i}P'_{\bsl k}$ contains additional terms involving $\partial_{k_i}W_{\bsl k}$. Consequently,
\eq{
g'_{ij}(\bsl{k}) = \frac{1}{2} \Tr\left[
\partial_{k_i}P'_{\bsl k}
\partial_{k_j}P'_{\bsl k}
\right]
}
generally differs from $g_{ij}(\bsl{k})$.
A uniform shift of all embedded positions $\bsl{\xi}_{\alpha}$ makes $W_{\bsl k}$ proportional to the identity matrix and therefore leaves the projector and its
quantum geometry unchanged. Relative shifts between different
orbital positions, however, can change the  quantum metric and Berry
curvature. The Chern number remains unchanged under a change of embedding.

\section{Isolated Chern-ideal bands in finite-band  tight-binding models with exponentially decaying hopping}
\label{sec:exp}

By \cref{thm:exp_decay_analyticity}, exponentially decaying hopping implies a real-analytic Bloch Hamiltonian and hence a real-analytic isolated-band projector.
Conversely, by \cref{thm:analytic_projector_parent_hamiltonian}, every real-analytic
rank-one projector satisfying $P_{\bsl k + \bsl G} = V_{\bsl{G}}P_{\bsl k}V_{\bsl{G}}^\dagger$ generates an exponentially decaying parent Hamiltonian. Therefore, it is sufficient in this section to study real-analytic ideal projectors with the appropriate embedding. The essential distinction is whether all embedded positions $\bsl{\xi}_{\alpha}$ coincide modulo the Bravais lattice, in which case $P_{\bsl k}$ is periodic under the shift of reciprocal lattice vectors (in short, $\bsl{G}$-periodic), or whether at least two embedded positions differ, in which case the projector obeys a nontrivial transformation rule under the shift of reciprocal lattice vectors.
We only focus on Chern-ideal bands with $|\Ch|=1$, as $|\Ch|>1$ Chern-ideal bands has been constructed in tight-binding models with exponentially decaying hopping \cite{Jian_2012,claassen2015position, Lee_2017,Mera_2021}.

\subsection{Review: Non-existence of isolated $|\Ch|=1$ Chern-ideal bands when all $\xi_{\alpha}$'s are the same modulo lattice vectors}
\label{app:no_ch_1_periodic}

The obstruction for $|\Ch|=1$ $\bsl{G}$-periodic ideal projectors with finite $N$ was previously formulated
in terms of momentum-space $\CP^{N-1}$ instantons and elliptic Bloch
functions, for which the Chern number is determined by the total pole
order in the BZ~\cite{Jian_2012,claassen2015position,Lee_2017,Jie2021IdealBands,Parker2023IdealBands,Mera_2021}.
We summarize the proof here.

Given an $N$-band tight-binding model with finite $N\geq2$ ($N=1$ cannot have $|\Ch|>0$ bands), a rank-$r$ projector lies in the Grassmannian $\Gr(r,N)$ \cite{Mera_2021}, which is defined as
\eq{
\Gr(r,N)
=
\left\{
P\in\dsC^{N\times N}
\mid
P^2=P,\quad
P^\dagger=P,\quad
\Tr[P]=r
\right\}.
}
A Grassmannian is a Kähler manifold, \ie, it has a complex structure $J$ that is compatible with its metric form $g$ and symplectic form $\omega$ \cite{Mera_2021}.
First, we note that the tangent space of $\Gr(r,N)$ at $P$ is
\eq{
T_P\Gr(r,N) = \left\{
\delta P \mid
\delta P^\dagger = \delta P,
P\delta P P = Q\delta P Q = 0
\right\},
}
where $Q=\mathbb{1}_N-P$.
Given two elements $\delta P_1$ and $\delta P_2$ in $T_P\Gr(r,N)$, the metric form is defined as
\eq{
g(\delta P_1,\delta P_2)
=
\Tr[\delta P_1\delta P_2],
}
and the symplectic form is defined as
\eq{
\omega(\delta P_1,\delta P_2)
=
\ii\Tr\left(P[\delta P_1,\delta P_2]\right).
}
The complex structure $J$ is defined as
\eq{
J(\delta P)
=
\ii[P,\delta P].
}
The intuitive way to understand this definition is the following.
Suppose $P=UU^\dagger$ and $Q=VV^\dagger$. Then
\eq{
\delta P
=
UAV^\dagger
+
VA^\dagger U^\dagger,
}
where $A$ is an $r\times(N-r)$ complex matrix.
Then
\eq{
J(\delta P) =
U(\ii A)V^\dagger +
V(\ii A)^\dagger U^\dagger,}
which means that $J$ is equivalent to $A\rightarrow\ii A$---exchanging the real and imaginary parts of $A$ up to a relative sign.
$J$ is compatible with $g$ and $\omega$ in the following sense:
\eqa{
J^2 & =  -1,\\
g(J(\delta P_1),J(\delta P_2)) & =  g(\delta P_1,\delta P_2), \\
g(J(\delta P_1),\delta P_2) &= \omega(\delta P_1, \delta P_2).
}
For the last line, 
\eq{
g(J(\delta P_1),\delta P_2) =
\ii\Tr\left(\delta P_2[P,\delta P_1]\right) =
\ii\Tr\left(P\delta P_1\delta P_2\right)
-\ii\Tr\left(P\delta P_2\delta P_1\right) =
\ii\Tr\left(P[\delta P_1,\delta P_2]\right)
= \omega(\delta P_1,\delta P_2).
}

For an isolated band of the tight-binding model with exponentially decaying hopping, its projector $P_{\bsl{k}}$ is a real-analytic map $\Phi$ from the two-torus $\BZ$ to $\Gr(1,N)$:
\eq{
\Phi(\bsl{k}) = P_{\bsl{k}}, 
}
since we have chosen all embedded positions to coincide modulo the Bravais lattice and thus $P_{\bsl{k}+\bsl{G}} = P_{\bsl{k}}$.
The map $\Phi$ pulls back the metric and symplectic forms to the two-torus.
At $P_{\bsl{k}}$, the two tangent vectors given by the pushforward are
\eq{
\delta P_1(\bsl{k}) = \partial_{k_x}P_{\bsl{k}},
\qquad
\delta P_2(\bsl{k}) = \partial_{k_y}P_{\bsl{k}}.
}
Thus, the pullbacks of the metric and symplectic forms are
\eq{
\left[\Phi^*g\right]_{ij}
= g\left(\partial_{k_i}P_{\bsl{k}},\partial_{k_j}P_{\bsl{k}}\right)
= \Tr\left[\partial_{k_i}P_{\bsl{k}}\partial_{k_j}P_{\bsl{k}}\right]
= 2\left[g_{\bsl{k}}\right]_{ij},}
and
\eq{
\left[\Phi^*\omega\right]_{ij}
= \omega\left(\partial_{k_i}P_{\bsl{k}},\partial_{k_j}P_{\bsl{k}}\right)
= \ii\Tr\left[
P_{\bsl{k}}
\left[
\partial_{k_i}P_{\bsl{k}},
\partial_{k_j}P_{\bsl{k}}
\right]\right]
= F_{ij}(\bsl{k}),
}
where $F_{\bsl{k}}\equiv F_{xy}(\bsl{k})$ is the Berry curvature.
The integration of the pulled-back symplectic form defines the degree of the map, which, with our convention, is the Chern number \cite{Jian_2012,Mera_2021}:
\eq{
\deg(\Phi) = \frac{1}{2\pi} \int_{\BZ}
\dd^2k F_{\bsl{k}} = \Ch.
}
As in the rest of the paper, let us focus on $\Ch>0$, so the ideal condition is
\eq{
F_{\bsl{k}}
=
\Tr[g_{\bsl{k}}].
}
By defining
\eq{
X_{\bsl{k} } = \partial_{k_x}P_{\bsl{k}} +J\left(\partial_{k_y}P_{\bsl{k}}\right) = \partial_{k_x}P_{\bsl{k}} +\ii [P_{\bsl{k}}, \partial_{k_y}P_{\bsl{k}}]
}
We have
\eqa{
\Tr[X_{\bsl{k} } X_{\bsl{k} } ]  & = \Tr[\partial_{k_x}P_{\bsl{k}}\partial_{k_x}P_{\bsl{k}}] + 2 \ii \Tr[ \partial_{k_x}P_{\bsl{k}} [P_{\bsl{k}}, \partial_{k_y}P_{\bsl{k}}] ] - \Tr[ [P_{\bsl{k}}, \partial_{k_y}P_{\bsl{k}}]    [P_{\bsl{k}}, \partial_{k_y}P_{\bsl{k}}]   ] \\
& = \Tr[\partial_{k_x}P_{\bsl{k}}\partial_{k_x}P_{\bsl{k}}] - 2 \ii \Tr[ P_{\bsl{k}} [ \partial_{k_x}P_{\bsl{k}}, \partial_{k_y}P_{\bsl{k}}] ] +2  \Tr[ P_{\bsl{k}} \partial_{k_y}P_{\bsl{k}} \partial_{k_y}P_{\bsl{k}} ] \\
& = \Tr[\partial_{k_x}P_{\bsl{k}}\partial_{k_x}P_{\bsl{k}}] - 2 \ii \Tr[ P_{\bsl{k}} [ \partial_{k_x}P_{\bsl{k}}, \partial_{k_y}P_{\bsl{k}}] ] +  \Tr[  \partial_{k_y}P_{\bsl{k}} \partial_{k_y}P_{\bsl{k}} ] \\
& = 2 \Tr[g_{\bsl{k}}] - 2 F_{\bsl{k}} = 0\ ,
}
which leads to
\eq{
\partial_{k_x}P_{\bsl{k}}
=
-J\left(\partial_{k_y}P_{\bsl{k}}\right) = - \ii [P_{\bsl{k}}, \partial_{k_y}P_{\bsl{k}}]\ .
}
Thus, with the complex momentum coordinate $k=k_x-\ii k_y$, the map $\Phi$ is $J$-holomorphic \cite{Jian_2012,Lee_2017,Jie2021IdealBands,Mera_2021,Parker2023IdealBands}.
(One can make an analogy to holomorphic functions in $\mathbb{C}$, $\partial_{k_x}\Phi(z)=-\ii\partial_{k_y}\Phi(z)$.)

Nonconstant holomorphic maps from a two-torus to 
$\Gr(1,N)\cong\CP^{N-1}$ cannot have degree one for any $N\geq2$ \cite{EellsWood1983,Jian_2012}.
We can understand this in the following way.
As mentioned above, let us only focus on $\Ch>0$, since the case $\Ch<0$ can be obtained by 
time-reversal.

Owing to $P_{\bsl{k}+\bsl{G}} = P_{\bsl{k}}$, we know $U_{\bsl{k}+\bsl{G}}=U_{\bsl{k}} e^{\ii \phi_{\bsl{k},\bsl{G}}}$ with $\phi_{\bsl{k},\bsl{G}}\in \dsR$.
Furthermore, there exists an atlas covering $\dsR^2$ such that $U_{\bsl{k}}$ is real-analytic in each patch.
Explicitly, for any $\bsl{k}_0$, we can choose an open neighborhood $\mathcal{U}$ on which $U_{\bsl{k}}$ is real-analytic.
In $\mathcal{U}$, the ideality gives
\eq{
\label{eq:partial_z_U}
(\partial_{k_x}-\ii\partial_{k_y})U_{\bsl{k}}
=
\lambda_{\bsl{k}}U_{\bsl{k}},
}
where $\lambda_{\bsl{k}}$ is a scalar.
It can be derived from
\eqa{
& \partial_{k_x}P_{\bsl{k}} = - \ii P_{\bsl{k}}  \partial_{k_y}P_{\bsl{k}} + \ii \partial_{k_y}P_{\bsl{k}} P_{\bsl{k}}  \\
& \Rightarrow (1-P_{\bsl{k}})\partial_{k_x}P_{\bsl{k}} =  \ii (1-P_{\bsl{k}}) \partial_{k_y}P_{\bsl{k}} P_{\bsl{k}} =  \ii (1-P_{\bsl{k}}) \partial_{k_y}P_{\bsl{k}} \\
&\Rightarrow (1-P_{\bsl{k}}) (\partial_{k_x}-\ii \partial_{k_y})P_{\bsl{k}} = 0\\
& \Rightarrow (1-P_{\bsl{k}}) (\partial_{k_x}-\ii \partial_{k_y})U_{\bsl{k}} = 0\ .
}
$\lambda_{\bsl{k}}$ can always be made nonzero by making $\mathcal{U}$ small enough, since $U_{\bsl{k}}\rightarrow U_{\bsl{k}}e^{\ii( a k_x + b k_y)}$ realizes $\lambda_{\bsl{k}} \rightarrow \lambda_{\bsl{k}} + \ii a + b$.

Now we show given a component $U_{{\bsl k},a}$, it can only be zero at isolated momentum points.
To see that, consider a generic zero $U_{{\bsl k}',a}=0$ at $\bsl{k}'$.
Then, we can always pick a small enough neighborhood such that another component $U_{{\bsl k},a'}\neq 0$ throughout the neighborhood.
Then, $U_{\bsl{k}}/U_{{\bsl k},a'}$ is holomorphic in the neighborhood. Consequently, the zero $U_{\bsl{k},a}/U_{{\bsl k},a'}$ at $\bsl{k}'$ must be isolated, as holomorphic functions can only have isolated zeros.
Since $U_{{\bsl k},a'}$ is nonzero in the neighborhood, the zero of $U_{\bsl{k},a}$ at $\bsl{k}_0$ is isolated.
Therefore, each component of $U_{\bsl{k}}$ can only have isolated zeros.
One can then choose one component of $U_{\bsl{k}}$, say $U_{\bsl{k},a}$, and obtain a meromorphic  function with 
\eq{w(k)=\frac{U_{\bsl{k}}}{U_{\bsl{k},a}}\ .
}
Owing to $U_{\bsl{k}+\bsl{G}}=U_{\bsl{k}} e^{\ii \phi_{\bsl{k},\bsl{G}}}$, $w(k)$ satisfies
\eq{
w(k+G) = w(k),
}
where
\eq{
k=k_x-\ii k_y, \quad G=G_x-\ii G_y.
}
Every component of $w(k)$ is meromorphic.

The K\"ahler potential on this K\"ahler manifold is given by $ w(k)^\dagger w(k)$. Then, let
\eq{
\beta_{\bsl{k}} = \frac{1}{\pi}
\partial_{k^*}\partial_k
\log\left[w(k)^\dagger w(k)\right].
}
In the domain of analyticity of $w(k)$
\eq{
\beta_{\bsl{k}}
= \frac{1}{\pi}\left[\frac{
\partial_{k^*}w(k)^\dagger
\partial_kv(k)
}{w(k)^\dagger w(k)
} - \frac{\left(\partial_{k^*}w(k)^\dagger w(k)\right)
\left(w(k)^\dagger\partial_kv(k)\right)
}{\left(w(k)^\dagger w(k)\right)^2
}\right].}
The Berry curvature reads

\eqa{
F_{\bsl{k}}
&= 
2\left[
\partial_{k^*}
\left(\frac{w(k)^\dagger}{\sqrt{w(k)^\dagger w(k)}}
\right) \partial_k
\left(\frac{w(k)}{\sqrt{w(k)^\dagger w(k)}}
\right) - \partial_k \left(
\frac{w(k)^\dagger}{\sqrt{w(k)^\dagger w(k)}}\right)\partial_{k^*}\left(\frac{w(k)}{\sqrt{w(k)^\dagger w(k)}}\right)\right]
\\
&= 2\left[\frac{\partial_{k^*}w(k)^\dagger\partial_kv(k)}{w(k)^\dagger w(k)} - \frac{\left(\partial_{k^*}w(k)^\dagger w(k)\right)
\left(w(k)^\dagger\partial_kv(k)\right)}{\left(w(k)^\dagger w(k)\right)^2}\right] = 2\pi\beta_{\bsl{k}}.}
Thus, if $k$ is not a pole of $w(k)$, then
\eq{
\beta_{\bsl{k}}
=
\frac{1}{2\pi}F_{\bsl{k}}.
}
If $k_i$ is a pole of $w(k)$, then in its neighborhood
\eq{
w(k)
=
(k-k_i)^{-q_i}m_i(k),
}
where $m_i(k)$ is holomorphic and $m_i(k_i)\neq0$.
Here $q_i$ is the maximal pole order among the components of $w(k)$ at $k_i$.
We have
\eqa{
\beta_{\bsl{k}}
&=
-q_i\frac{1}{\pi}
\partial_{k^*}\partial_k
\log|k-k_i|^2
+
\text{regular term}
\\
&=
-q_i\frac{1}{\pi}
\partial_{k^*}
\frac{1}{k-k_i}
+
\text{regular term}
\\
&=
-q_i\delta(\bsl{k}-\bsl{k}_i)
+
\text{regular term}.
}
Therefore,
\eq{\beta_{\bsl{k}} = \frac{1}{2\pi}F_{\bsl{k}} - \sum_iq_i\delta(\bsl{k}-\bsl{k}_i).}
The integration of $\beta_{\bsl{k}}$ over the first BZ is a total derivative of a periodic distribution on the BZ, and thus must be zero. Thus,
\eq{\Ch = \sum_iq_i.\label{eq:chern_number_pole_count}
}
The total pole order $\sum_iq_i$ cannot equal one, since a nonconstant elliptic function (\ie, doubly periodic meromorphic functions) cannot have exactly one simple pole in a fundamental domain~\cite{StoneGoldbart2009}. Therefore,
\eq{
\Ch = \sum_iq_i \geq 2.
}
Note that it is the total pole order $\sum_iq_i$, rather than every individual $q_i$, that must be no smaller than two.
Therefore, a Chern-ideal band with $\Ch>0$ in a tight-binding model with exponentially decayed hopping must have $\Ch>1$.
In other words, no $\Ch=1$ Chern-ideal band in this case.
Using the time-reversal operation, we know $\Ch=-1$ Chern-ideal bands also do not exist in this case.

Explicit constructions of $|\Ch|>1$ Chern-ideal bands have been obtained~\cite{Jian_2012,claassen2015position, Lee_2017,Mera_2021}.

\subsection{Existence and analytic construction of isolated $|\Ch|=1$ Chern-ideal bands when at least two $\xi_{\alpha}$'s are different}\label{subsub: construction for f at diff positions}

In the previous section, we considered a periodically embedded two-band Hilbert space, for which $\bsl{\xi}_{\alpha}$'s are all lattice vectors, so $P_{\bsl k + \bsl G}=P_{\bsl k}$.
This periodicity was the essential global assumption behind the nonexistence result for $|\mathrm{Ch}|=1$.
In this section, we instead consider the embedding where at least two $\xi_{\alpha}$'s are different modulo lattice vectors. The local relation between ideality and holomorphicity remains unchanged, but the transformation rule of the meromorphic vector is different under the shift of reciprocal lattice vectors. We explicitly construct a meromorphic function $f$ with exactly one simple pole and show that it defines an ideal band with $\Ch=1$, and its time-reversal partner gives $\Ch=-1$.

Let $\bsl a_1$ and $\bsl a_2$ be primitive vectors of a
two-dimensional Bravais lattice $\mathcal L
=
\left\{
m\bsl a_1+n\bsl a_2:
m,n\in\mathbb Z
\right\}.$
Let $\bsl b_1$ and $\bsl b_2$ be the corresponding primitive
reciprocal vectors, chosen such that $\bsl a_i\cdot\bsl b_j=2\pi\delta_{ij}$.
We choose the embedded positions
$\boldsymbol\xi_A$ and $\boldsymbol\xi_B$ and assume that
\eq{
\Delta\boldsymbol\xi
\equiv
\boldsymbol\xi_B-\boldsymbol\xi_A
\notin\mathcal L.
\label{eq:nontrivial_relative_embedding}}
We use the orbital Bloch basis, and the eigenvector of the corresponding isolated band reads
\eq{
U_{\bsl k}
=
\begin{pmatrix}
U_A(\bsl k)
\\
U_B(\bsl k)
\end{pmatrix}.
\label{eq:orbital_coefficient_spinor}
}
Owing to $P_{\bsl{k}+\bsl{G}} = V_{\bsl{G}}P_{\bsl{k}}V_{\bsl{G}}^\dagger$, we know
\eq{U_\alpha(\bsl k+\bsl G) =
e^{-i\bsl G\cdot\bsl\xi_\alpha}
U_\alpha(\bsl k) e^{\ii \phi_{\bsl{k},\bsl{G}}} \Leftrightarrow U_{\bsl k+\bsl G}
= V_{\bsl G}U_{\bsl k} e^{\ii \phi_{\bsl{k},\bsl{G}}}\ ,
\label{eq:component_coefficient_transformation}} where
$V_{\bsl G}
=
\begin{pmatrix}
e^{-i\bsl G\cdot\bsl\xi_A} & 0
\\
0 & e^{-i\bsl G\cdot\bsl\xi_B}
\end{pmatrix}.$
Let us now define
\eq{
w_B(\bsl k)
=
\frac{U_B(\bsl k)}{U_A(\bsl k)}\ ,
\label{eq:projective_coordinate_definition}
}
leading to
\eq{
U_{\bsl k}
=
\frac{1}{\sqrt{1+|w_B(\bsl k)|^2}}
\begin{pmatrix}
1
\\
w_B(\bsl k)
\end{pmatrix}.
\label{eq:normalized_spinor_from_f}
}
and
\eq{
P_{\bsl k} =
U_{\bsl k}U^\dagger_{\bsl k} =
\frac{1}{1+|w_B(\bsl k)|^2}
\begin{pmatrix}
1 & w_B(\bsl k)^*
\\
w_B(\bsl k) & |w_B(\bsl k)|^2
\end{pmatrix}.
\label{eq:projector_from_projective_coordinate}
}
From
\cref{eq:component_coefficient_transformation},
\eq{
w_B(\bsl k+\bsl G) =
\frac{
U_B(\bsl k+\bsl G)
}{
U_A(\bsl k+\bsl G)
} =
\frac{
e^{\ii \phi_{\bsl{k},\bsl{G}}}
e^{-i\bsl G\cdot\bsl\xi_B}
U_B(\bsl k)
}{
e^{\ii \phi_{\bsl{k},\bsl{G}}}
e^{-i\bsl G\cdot\bsl\xi_A}
U_A(\bsl k)
} =
e^{-i\bsl G\cdot
(\bsl\xi_B-\bsl\xi_A)}
w_B(\bsl k).
\label{eq:physical_embedding_f_transformation}
}
Write
\eq{\Delta\bsl{\xi}
\equiv
\bsl{\xi}_B-\bsl{\xi}_A
= \nu_1\bsl{a}_1+\nu_2\bsl{a}_2
\pmod{\mathcal{L}},
\qquad
\nu_1,\nu_2\in[0,1).}
Since $\bsl{b}_j\cdot\bsl{a}_i = 2\pi\delta_{ij},$
we get $\bsl{b}_1\cdot\Delta\bsl{\xi}
= 2\pi\nu_1
\pmod{2\pi}, \quad \bsl{b}_2\cdot\Delta\bsl{\xi}
= 2\pi\nu_2 \pmod{2\pi}.$
Therefore, let
\eq{ \eta_j \equiv e^{-i\bsl{b}_j\cdot\Delta\bsl{\xi}}
= e^{-2\pi i\nu_j}, \qquad j=1,2.
\label{eq:eta_j def}
}
So
\begin{align}
V_{\bsl{b}_j}
&=
\begin{pmatrix}
e^{-i\bsl{b}_j\cdot\bsl{\xi}_A} & 0
\\
0 & e^{-i\bsl{b}_j\cdot\bsl{\xi}_B}
\end{pmatrix} =
e^{-i\bsl{b}_j\cdot\bsl{\xi}_A}
\begin{pmatrix}
1 & 0
\\
0 & \eta_j
\end{pmatrix}.
\end{align}
Using \cref{eq:eta_j def}, \cref{eq:physical_embedding_f_transformation} therefore becomes
\eq{
w_B(\bsl k+\bsl b_1) = \eta_{1} w_B(\bsl k),
\qquad w_B(\bsl k+\bsl b_2)
= \eta_{2} w_B(\bsl k).
\label{eq:honeycomb_f_quasiperiodicity}}

Now we parameterize the momentum by a complex number. Let
\eq{\bsl{k} = \kappa_1\bsl{b}_1 + \kappa_2\bsl{b}_2,
\qquad \kappa_1,\kappa_2\in\mathbb{R}.
\label{eq:momentum_reciprocal_coordinates}}
Define
\eq{
k = k_x-\ii k_y,
\qquad B_1 = b_{1x}-\ii b_{1y},
\qquad
B_2
= b_{2x}-\ii b_{2y}.
\label{eq:complex_representatives}
}
Applying this complex representation to
\cref{eq:momentum_reciprocal_coordinates} gives
\eq{
k
=
k_x-\ii k_y
=
\kappa_1B_1+\kappa_2B_2.
\label{eq:complex_momentum_decomposition}
}
Define the normalized complex momentum coordinate
\eq{
z = \frac{k}{B_1}
= \frac{\kappa_1B_1+\kappa_2B_2}{B_1}
= \kappa_1
+ \kappa_2\frac{B_2}{B_1}
= \kappa_1+\kappa_2\omega,
\label{eq:normalized_complex_momentum}
}
where $\omega\equiv \frac{B_2}{B_1} = \frac{b_{2x}-\ii b_{2y}}{b_{1x}-\ii b_{1y}}$. We order $\bsl{b}_1$ and $\bsl{b}_2$ such that $\operatorname{Im}\omega>0$. To do so, notice that
\eq{
\omega =
\frac{b_{2x}-\ii b_{2y}}
     {b_{1x}-\ii b_{1y}} =
\frac{
(b_{2x}-\ii b_{2y})
(b_{1x}+\ii b_{1y})
}{
b_{1x}^2+b_{1y}^2
} =
\frac{b_{1x}b_{2x}+b_{1y}b_{2y}}{|\bsl{b}_1|^2}
- \ii\frac{b_{1x}b_{2y}-b_{1y}b_{2x}}{|\bsl{b}_1|^2}.}
Therefore, $\operatorname{Im}\omega
= -\frac{\det(\bsl{b}_1,\bsl{b}_2)}{|\bsl{b}_1|^2}$.
Since $\bsl{b}_1$ and $\bsl{b}_2$ are linearly independent,
$\det(\bsl{b}_1,\bsl{b}_2)\neq0$. We can order them such
that $\det(\bsl{b}_1,\bsl{b}_2)<0,$ so that $\operatorname{Im}\omega>0$.

Our construction is done by finding $f(z)$ for
\eq{w_B(\bsl{k})  = f(z) \equiv w_B(k) ,}
so that we have
\eq{ w(k) = \mat{1 \\ w_B(k) } = \mat{1 \\ f(z)}}
for $w(k)$ in \cref{eq:w_k}.
In the complex coordinate $z$, the first reciprocal
vector $\bsl b_1$ is represented by the number $1$, while
the second reciprocal vector $\bsl b_2$ is represented by
$\omega$. Indeed, shifting the momentum by $\bsl b_1$ changes
$\kappa_1\to\kappa_1+1$ while leaving $\kappa_2$ unchanged.
Hence,
\eq{
\bsl{k}\longmapsto\bsl{k}+\bsl{b}_1\quad\Longleftrightarrow\quad
z\longmapsto z+1;}
\eq{
\bsl{k}\longmapsto\bsl{k}+\bsl{b}_2
\quad\Longleftrightarrow\quad
z\longmapsto z+\omega.
}
Thus, Bloch momenta that differ by a reciprocal-lattice vector are
identified by $z \sim z+m+n\omega$, with $m,n\in\mathbb Z$.
The BZ is thus represented by the complex torus $\mathbb C/\Lambda$ with $\Lambda = \mathbb Z+\omega\mathbb Z.$
Therefore, the conditions are
\eq{f(z+1)=\eta_{1} f(z),
\qquad f(z+\omega)=\eta_{2} f(z).
\label{eq:honeycomb_f_complex_quasiperiodicity}}
Unlike the periodically embedded case, now $f$ does not return to the same value after a reciprocal-lattice translation.

In order to construct an isolated ideal band with
$\mathrm{Ch}=1$, we will show that it is sufficient to find a function $f$ satisfying:
\eq{
\boxed{
\begin{aligned}
&f:\mathbb C\longrightarrow\mathbb C\cup\{\infty\}
\text{ is meromorphic}.\\
&f(z+1)=\eta_{1} f(z)\quad f(z+\omega)=\eta_{2} f(z).\\
&f(z)\text{ has exactly one simple pole modulo }\Lambda.
\end{aligned}
}
\label{eq:quasiperiodic_meromorphic_conditions}
}
$P_{\bsl{k}+\bsl{G}} = V_{\bsl{G}} P_{\bsl{k}} V_{\bsl{G}}^\dagger$ is clearly satisfied due to the second condition, and the pole-counting relation in \cref{eq:chern_number_pole_count} ensures $\Ch=1$ due to the third condition. Now we give explicit proof on how the conditions ensure that the ideal condition is satisfied.

We first show the real analyticity of the associated
projector. 
Away from the poles of $f$, the function $f(z)$ is holomorphic and hence real-analytic in the real momentum coordinates $(k_x,k_y)$. Therefore, $P_{\bsl k}$ is real-analytic away from the poles, since $1+|f(z)|^2>0$. Therefore, it remains
to examine the behavior of $P_{\bsl k}$ at the simple pole.
Around a simple pole at generic $z=p$, there exists a nonzero constant
$c_p$ such that $f(z)=\frac{c_p + v(z-p)}{z-p}$ with $v(0) = 0 $.
Near $p$, introduce $\widetilde{f}(z) \equiv \frac{1}{f(z)} = \frac{z-p}{c_p+v(z-p)}.$
Therefore, $\widetilde{f}$ is holomorphic on a neighborhood of $p$
and satisfies $\widetilde{f}(p)=0.$
In terms of $\widetilde{f}$, the band projector is
\eq{
\label{eq:P_ftilde}
P_{\bsl{k}}=\frac{1}{1+|\widetilde{f}(z)|^2}
\begin{pmatrix}
|\widetilde{f}(z)|^2 & \widetilde{f}(z)\\
\widetilde{f}^*(z) & 1
\end{pmatrix}\ .
}
Since $\widetilde{f}(z)$ is holomorphic, both $\widetilde{f}(z)$ and $\widetilde{f}(z)^*$ are
real-analytic in $(k_x,k_y)$. 
So pole of $f=U_B/U_A$ is not a singularity of the band projector. Moreover,
$1+|\widetilde{f}(z)|^2>0$, so every matrix element of
$P_{\bsl{k}}$ is real-analytic near $p$.
In particular, at $p$, the projector takes the form $\begin{pmatrix}
0 & 0
\\
0 & 1
\end{pmatrix}$
.
Since the same argument applies at every lattice-translated pole, $P_{\bsl{k}}$ is real-analytic on all of $\dsR^2$.
We next verify that $\Tr[g_{\bsl{k}}]=F_{\bsl{k}}$
for every $\bsl{k}\in\dsR^2$.
Let
$Q_{\bsl{k}}=\mathbb{1}_2-P_{\bsl{k}}$.
Away from the pole of $f$, we have $\Tr[g_{\bsl{k}}]-F_{\bsl{k}} = \norm{
Q_{\bsl{k}} \left(\partial_{k_x}-\ii\partial_{k_y}
\right) U_{\bsl{k}}
}^2=0$.
Since $\Tr[g_{\bsl{k}}]-F_{\bsl{k}}$ is real-analytic in $\mathbb{R}^2$, $\Tr[g_{\bsl{k}}]-F_{\bsl{k}} = 0$ should also hold at the poles of $f$.

In the following, we will show that the following ansatz of $f$  satisfies the conditions in \cref{eq:quasiperiodic_meromorphic_conditions} with a proper choice of $a$ and $\delta$:
\eq{
f(z) = e^{az} \frac{
\Theta(z-p-\delta\mid\omega)
}{ \Theta(z-p\mid\omega) },
\label{eq:theta_function_ansatz}}
where $p$ determines the positions of the poles and $p+\delta$ the positions of the zeros.
Here $\Theta(z\mid\omega)$ denotes the odd Jacobi theta function~\cite{Kharchev_2015}:
\eq{\Theta(z\mid\omega)
=2\sum_{n=0}^{\infty}(-1)^n
e^{\pi i\omega\left(n+\frac{1}{2}\right)^2}\sin\left[(2n+1)\pi z\right]\ ,\label{eq:Jacobi_1}}
and an equivalent alternative form is
\eq{
\Theta(z\mid\omega) = -i
\sum_{n\in\mathbb Z}
(-1)^n e^{i\pi\omega\left(n+\frac12\right)^2}
e^{i2\pi\left(n+\frac12\right)z}.
\label{eq:Jacobi_2}
}

\subsubsection{Review on the odd Jacobi theta function}

We first review the relevant properties of the odd Jacobi theta function \cite{Haldane_1985_LJ_function,Kharchev_2015,Slavnov_2020,kulkarni2023formfactorlocaloperators}. Readers familiar with them may skip directly to the verification of the
ansatz. In our convention, $\Theta(z\mid\omega)$ is
entire in $z$ and satisfies
\eq{
\Theta(z+1\mid\omega)
= -\Theta(z\mid\omega),
\qquad \Theta(z+\omega\mid\omega) = -e^{-\pi\ii\omega-2\pi\ii z}
\Theta(z\mid\omega),
\label{eq:odd_theta_shift}
}
and has exactly one simple zero at $z=0$ modulo
$\Lambda=\mathbb Z+\omega\mathbb Z$.
For completeness, we verify these standard properties now.

To see the equivalence of \cref{eq:Jacobi_2}, we note
\eqa{
\Theta(z\mid\omega)
&= 2\sum_{n=0}^{\infty}
(-1)^n
e^{\pi i\omega\left(n+\frac12\right)^2}
\sin\left[(2n+1)\pi z\right] =
2\sum_{n=0}^{\infty}
(-1)^n
e^{\pi i\omega\left(n+\frac12\right)^2}
\frac{
e^{2\pi i\left(n+\frac12\right)z}
- e^{-2\pi i\left(n+\frac12\right)z}
}{2i} \\
&=
-i\sum_{n=0}^{\infty}
(-1)^n
e^{\pi i\omega\left(n+\frac12\right)^2}
e^{2\pi i\left(n+\frac12\right)z}
+i\sum_{n=0}^{\infty}
(-1)^n
e^{\pi i\omega\left(n+\frac12\right)^2}
e^{-2\pi i\left(n+\frac12\right)z}. \\
&=
-i\sum_{m=0}^{\infty}
(-1)^m
e^{\pi i\omega\left(m+\frac12\right)^2}
e^{2\pi i\left(m+\frac12\right)z}
-i\sum_{m=-\infty}^{-1}
(-1)^m
e^{\pi i\omega\left(m+\frac12\right)^2}
e^{2\pi i\left(m+\frac12\right)z}\\
&=
-i\sum_{m\in\mathbb Z}
(-1)^m
e^{\pi i\omega\left(m+\frac12\right)^2}
e^{2\pi i\left(m+\frac12\right)z}.
\label{eq:theta_bilateral_series}
}

To show $\Theta(z \mid \omega )$ is entire in $z$ (\ie, holomorphic for every $z\in\mathbb C$), we  consider the $n$th term of the series
\eq{T_n(z) = 2(-1)^n
e^{\pi i\omega\left(n+\frac12\right)^2}
\sin \left[(2n+1)\pi z\right].
\label{eq:theta_nth_term}}
For a fixed $n$, the factor
$e^{\pi i\omega(n+\frac12)^2}$ is independent of $z$, while
$\sin[(2n+1)\pi z]$ is entire in $z$. Therefore, each $T_n(z)$ is
entire. 
To show that the infinite sum is also entire, it is sufficient to prove that the series \cref{eq:Jacobi_1}  converges locally uniformly, which we show below.
Let $\omega=\omega_R+i\omega_I $, $\omega_I>0$.
First,
$ \left|e^{\pi i\omega\left(n+\frac12\right)^2}\right|=
e^{-\pi\omega_I\left(n+\frac12\right)^2}.$
Next, consider the convergence of the series on any compact subset $K\subset\mathbb{C}$.
Since $K$ is compact, there is a supremum $C_K=\sup_{z\in K} |\Im z|$.
Then, for every $z\in K$ 
we have $\left| \sin \left[(2n+1)\pi z\right] \right| \le e^{(2n+1)\pi C_K}$ which is $z$ independent.
Therefore,
\eq{|T_n(z)| \le 2e^{-\pi\omega_I\left(n+\frac12\right)^2} e^{(2n+1)\pi C_K}\equiv M_n.}
It remains to show that the
$\sum_{n=0}^{\infty}M_n$ converges. We take the ratio
\eq{\frac{M_{n+1}}{M_n} = \exp\left[-2\pi\omega_I(n+1)+2\pi C_K \right].}
Since $\omega_I>0$, $\lim_{n\to\infty}
\frac{M_{n+1}}{M_n} = 0,$ which is smaller than 1. Hence $\sum_{n=0}^{\infty}M_n$ converges absolutely by the ratio test. Since $|T_n(z)|\le M_n$ for $z\in K$, it follows that
$\sum_{n=0}^{\infty}T_n(z)$ converges uniformly in the compact subset $K$.
Since $K$ is chosen arbitrarily, we know the series converges locally uniformly.
Weierstrass's theorem implies that given a sequence of entire functions such that (i) any finite partial sum is entire and (ii) the sequence converges locally uniformly, the sequence is entire. 
Therefore, $\Theta(z\mid\omega)$ is entire, which means $\Theta(z\mid\omega)$ does not diverge for any $z\in\dsC$.

To see \cref{eq:odd_theta_shift} explicitly, we note
\eqa{\Theta(z+1\mid\omega)
&=
2\sum_{n=0}^{\infty}
(-1)^n
e^{\pi i\omega\left(n+\frac12\right)^2}
\sin \left((2n+1)\pi(z+1)\right) =
2\sum_{n=0}^{\infty}
(-1)^n
e^{\pi i\omega\left(n+\frac12\right)^2}
\sin \left((2n+1)\pi z+(2n+1)\pi\right)\\
&= 
-2\sum_{n=0}^{\infty}
(-1)^n
e^{\pi i\omega\left(n+\frac12\right)^2} \sin\left[(2n+1)\pi z\right] =
-\Theta(z\mid\omega),
}
and
\eqa{
\Theta(z+\omega\mid\omega)
&=
-i \sum_{n\in\mathbb Z}
(-1)^n e^{\pi i\omega\left(n+\frac12\right)^2}
e^{2\pi i\left(n+\frac12\right)(z+\omega)} = -i\sum_{n\in\mathbb Z} (-1)^n e^{\pi i\omega\left(n+\frac12\right)^2} e^{2\pi i\left(n+\frac12\right)z} e^{2\pi i\omega\left(n+\frac12\right)}
\\
&=
-i \sum_{n\in\mathbb Z}
(-1)^n e^{\pi i\omega
\left[ \left(n+\frac32\right)^2-1 \right]} e^{2\pi i\left(n+\frac12\right)z} = -i e^{-\pi i\omega}\sum_{n\in\mathbb Z} (-1)^n e^{\pi i\omega\left(n+\frac32\right)^2} e^{2\pi i\left(n+\frac12\right)z}
\\
&=
-i e^{-\pi i\omega}
\sum_{m\in\mathbb Z}
(-1)^{m-1} e^{\pi i\omega\left(m+\frac12\right)^2}
e^{2\pi i\left(m-\frac12\right)z}
\\
&=
i e^{-\pi i\omega}
\sum_{m\in\mathbb Z}
(-1)^m
e^{\pi i\omega\left(m+\frac12\right)^2} e^{-2\pi iz} e^{2\pi i\left(m+\frac12\right)z} = i e^{-\pi i\omega-2\pi iz}
\sum_{m\in\mathbb Z} (-1)^m e^{\pi i\omega\left(m+\frac12\right)^2}
e^{2\pi i\left(m+\frac12\right)z}
\\
&=
-e^{-\pi i\omega-2\pi iz}
\Theta(z\mid\omega).
}

We now show that $\Theta$ has exactly one zero at $z=0$ modulo the lattice. Let $q=e^{\pi i\omega}$ with $\operatorname{Im}\omega>0$, then $|q|<1$.
\Cref{eq:Jacobi_2} becomes
\eq{
\Theta(z\mid\omega) =
-iq^{1/4}e^{\pi iz}
\sum_{n\in\mathbb Z}
(-1)^n
q^{n^2+n}
e^{2\pi inz} =
-iq^{1/4}e^{\pi iz}
\sum_{n\in\mathbb Z}
q^{n^2}
\left(-q e^{2\pi iz}\right)^n.
}
The standard Jacobi triple-product identity is \cite{Jacobi1829}
\eq{
\sum_{n\in\mathbb Z}
x^{n^2}y^{2n}
=
\prod_{m=1}^{\infty}
\left(1-x^{2m}\right)
\left(1+x^{2m-1}y^2\right)
\left(1+x^{2m-1}y^{-2}\right).
}
Choose $x=q$ and $y^2=-q e^{2\pi iz}$ here, then
\eq{
\sum_{n\in\mathbb Z}
q^{n^2}
\left(-q e^{2\pi iz}\right)^n
=
\prod_{m=1}^{\infty}
\left(1-q^{2m}\right)
\left(1-q^{2m}e^{2\pi iz}\right)
\left(1-q^{2m-2}e^{-2\pi iz}\right).
}
Separate the $m=1$ term from the last product:
\eq{
\Theta(z\mid\omega) =
-iq^{1/4}e^{\pi iz}
\left(1-e^{-2\pi iz}\right) \times
\prod_{m=1}^{\infty}
\left(1-q^{2m}\right)
\left(1-q^{2m}e^{2\pi iz}\right)
\left(1-q^{2m}e^{-2\pi iz}\right).}
Since $-i e^{\pi iz}
\left(1-e^{-2\pi iz}\right) =
-i\left(e^{\pi iz}-e^{-\pi iz}\right) =
2\sin(\pi z),$
\eq{
\Theta(z\mid\omega)
=
2q^{1/4}\sin(\pi z)
\prod_{m=1}^{\infty}
\left(1-q^{2m}\right)
\left(1-q^{2m}e^{2\pi iz}\right)
\left(1-q^{2m}e^{-2\pi iz}\right),
\qquad
q=e^{\pi i\omega}.
\label{eq:produc_representation}
}
Because $|q|<1$, the infinite products in
\cref{eq:produc_representation} converge locally uniformly in $z$. Indeed, for any compact set $K\subset\mathbb C$, let
$C_K=\sup_{z\in K}|\operatorname{Im}z|$. Then
\eq{ \sum_{m=1}^{\infty}
\sup_{z\in K}
\left| q^{2m}e^{\pm2\pi\ii z} \right|
\leq e^{2\pi C_K}
\sum_{m=1}^{\infty}|q|^{2m}
<\infty.} 
Moreover, $\sum_{m=1}^{\infty}|q|^{2m}<\infty$.
Therefore, by the standard convergence theorem for infinite products, the products converge locally uniformly to holomorphic functions and are nonzero wherever none of their individual factors vanishes. Hence the infinite product introduces no additional zeros. The zeros of $\Theta(z\mid\omega)$ therefore arise from one of the $z$-dependent product factors.
For the first factor,
\eq{\sin(\pi z)=0 \quad\Longleftrightarrow\quad z=\ell, \qquad \ell\in\mathbb Z.}
For the second one,
\eqa{1-q^{2m}e^{2\pi iz}=0
&\quad\Longleftrightarrow\quad
e^{2\pi im\omega}e^{2\pi iz}=1
\\
&\quad\Longleftrightarrow\quad
e^{2\pi i(z+m\omega)}=1
\\
&\quad\Longleftrightarrow\quad
z=\ell-m\omega,
\qquad
\ell\in\mathbb Z,
\quad
m\in\mathbb N.
}
Lastly, for the third one,
\eqa{
1-q^{2m}e^{-2\pi iz}=0
&\quad\Longleftrightarrow\quad
e^{2\pi im\omega}e^{-2\pi iz}=1
\\
&\quad\Longleftrightarrow\quad
e^{2\pi i(m\omega-z)}=1
\\
&\quad\Longleftrightarrow\quad
z=m\omega-\ell,
\qquad
\ell\in\mathbb Z,
\quad
m\in\mathbb N.
}
Combining the three cases, the zero set is
\eq{
\Theta(z\mid\omega)=0
\quad\Longleftrightarrow\quad
z=\ell+m\omega,
\qquad
\ell,m\in\mathbb Z.
}
Equivalently,
\eq{
\Theta(z\mid\omega)=0
\quad\Longleftrightarrow\quad
z\in\Lambda,
\qquad
\Lambda=\mathbb Z+\omega\mathbb Z.
}
Therefore, $\Theta(z\mid\omega)$ has exactly one zero modulo $\Lambda$. 

Next, we show that each zero of $\Theta$ is simple.
To see it, define
\eq{\mathcal R(z)\equiv \prod_{m=1}^{\infty}
\left(1-q^{2m}\right) \left(1-q^{2m}e^{2\pi iz}\right)
\left(1-q^{2m}e^{-2\pi iz}\right), \qquad |q|<1,}
so \cref{eq:produc_representation}
becomes
\eq{
\Theta(z\mid\omega) =
2q^{1/4}\sin(\pi z)\mathcal R(z).
}
At $z=0$, we have $\Theta(0\mid\omega)=0$
and $\mathcal R(0)= \prod_{m=1}^{\infty}\left(1-q^{2m}\right)^3 \neq0$, by the convergence and nonvanishing argument above.
Differentiating,
\eq{
\Theta'(z\mid\omega) =
2q^{1/4}
\left[\pi\cos(\pi z)\mathcal R(z) + \sin(\pi z)\mathcal R'(z)\right],
}
hence $\Theta'(0\mid\omega)
= 2\pi q^{1/4}\mathcal R(0)\neq0.$ Therefore,
\eq{
\Theta(0\mid\omega)=0,\qquad \Theta'(0\mid\omega)\neq0,
}
so $z=0$ is a simple zero of $\Theta(z\mid\omega)$.

\subsubsection{Verification of the Ansatze}

We now show that the ansatze \cref{eq:theta_function_ansatz} with proper $a$ and $\delta$ satisfies the conditions in  \cref{eq:quasiperiodic_meromorphic_conditions}.

We first determine $a$ and $\delta$ in $f(z)$, using $f(z+1)=\eta_1 f(z)$ and $f(z+\omega)=\eta_2 f(z)$.
First, using Eq.~\eqref{eq:odd_theta_shift},
\eq{
f(z+1) = e^{a(z+1)}
\frac{ -\Theta(z-p-\delta\mid\omega)
}{-\Theta(z-p\mid\omega)} = e^a f(z).
\label{eq:f_shift_one_general}}
We want $f(z+1)=\eta_1 f(z)$, so choose $a=-2\pi i\nu_1.$ Second, using Eq.~\eqref{eq:odd_theta_shift} again,
\eq{
f(z+\omega)
= e^{a(z+\omega)}
\frac{-e^{-\pi i\omega-2\pi i(z-p-\delta)}\Theta(z-p-\delta\mid\omega)}
{-e^{-\pi i\omega-2\pi i(z-p)}\Theta(z-p\mid\omega)}
= e^{a\omega}e^{2\pi i\delta}e^{az}
\frac{\Theta(z-p-\delta\mid\omega)}
{\Theta(z-p\mid\omega)}
= e^{a\omega}e^{2\pi i\delta}f(z)}
To obtain $f(z+\omega)=\eta_2 f(z)$, choose $\delta
=\nu_1 \omega - \nu_2.$ Indeed,
\eq{e^{a\omega}e^{2\pi i\delta} =
e^{-2\pi i\nu_1\omega}
e^{2\pi i(\nu_1\omega-\nu_2)} =
e^{-2\pi i\nu_2} = \eta_2.\label{eq:delta_gives_eta}}
We have thus obtained the explicit function
\eq{ f(z) =
e^{-2\pi i\nu_1 z}
\frac{
\Theta\left(z-p-\nu_1\omega+\nu_2\mid\omega
\right)}{\Theta\left(z-p\mid
\omega\right)}.\label{eq:explicit_quasiperiodic_f}}
So by construction, it satisfies the condition in \cref{eq:honeycomb_f_complex_quasiperiodicity}, \ie, the second condition in \cref{eq:quasiperiodic_meromorphic_conditions}.

\medskip

We now show that $f$ has only one simple pole modulo $\Lambda$ [the third condition of \cref{eq:quasiperiodic_meromorphic_conditions}].
Recall that $\Theta(z \mid \omega )$ is entire in $z$ with exactly one simple zero at $z=0$ modulo $\Lambda$. 
As a result, for \cref{eq:explicit_quasiperiodic_f},  the denominator gives 
simple poles at $z=p + \dsZ + \omega \dsZ$, whereas the numerator gives simple zeros at $z=p+\delta + \dsZ + \omega \dsZ$ with $\delta=\nu_1\omega-\nu_2$.
Suppose, for contradiction, that
\eq{
\nu_1\omega-\nu_2
= m+n\omega, \qquad
m,n\in\mathbb Z.
}
Comparing their coefficients therefore gives $\nu_1=n$ and $-\nu_2=m$.
Because $\nu_1,\nu_2\in[0,1)$ and $m,n\in\mathbb Z$, this is
possible only if $\nu_1=\nu_2=0$, which means $\Delta\bsl{\xi}\in\mathcal{L}$, contrary to our embedding assumption in \cref{eq:nontrivial_relative_embedding}. Therefore, the simple zeros and poles of $f(z)$ are distinct.
Thus, clearly  $f(z)$ has only one simple pole modulo $\Lambda$.
Therefore, $f$ satisfies all three conditions in \cref{eq:quasiperiodic_meromorphic_conditions}, serving as the construction of an isolated $\Ch=1$ Chern-ideal bands in a 2-band model with exponentially decaying hopping.

Since the Berry curvature has to have zeros for an isolated band of a 2-band model \cite{Mera_2021}, our ideal construction should have zeros for $\Tr[g_{\boldsymbol{k}}]$.
Let us verify this.
From the definition of the quantum metric,
\eq{
\Tr[g_{\boldsymbol{k}}]=0
\quad\Longleftrightarrow\quad
\partial_{k_x}P_{\bsl{k}}=\partial_{k_y}P_{\bsl{k}}=0.
\label{eq:metric_zero_real_derivatives}
}
Recall that $z=\frac{k_x-\ii k_y}{B_1}$, $B_1=b_{1x}-\ii b_{1y}.$
The corresponding complex derivatives are $\partial_z
=
\frac{B_1}{2} \left(\partial_{k_x}+\ii\partial_{k_y}
\right)$, and $\partial_{z^*}
= \frac{B_1^*}{2}
\left(\partial_{k_x}-\ii\partial_{k_y}
\right)$.
Since this is an invertible linear change of derivatives,
\eq{
\partial_{k_x}P_{\bsl{k}}=\partial_{k_y}P_{\bsl{k}}=0
\quad\Longleftrightarrow\quad
\partial_zP_{\bsl{k}}=\partial_{z^*}P_{\bsl{k}}=0.
}
Moreover, because $P_{\bsl{k}}^\dagger=P_{\bsl{k}}$, $\left(\partial_zP_{\bsl{k}}\right)^\dagger
= \partial_{z^*}P_{\bsl{k}}.$ Thus,
\eq{
\Tr[g_{\boldsymbol{k}}]=0
\quad\Longleftrightarrow\quad
\partial_zP_{\bsl{k}}=0.
\label{eq:metric_zero_dzP}
}
Away from the pole of $f$, $P_{\bsl{k}}$ from \cref{eq:projector_from_projective_coordinate}
and $\partial_zf(z)^*=0$, we obtain
\eq{
\partial_z P_{\bsl{k}}
= \frac{\partial_zf(z)}{(1+|f(z)|^2)^2}
\begin{pmatrix}
-f(z)^* & -(f(z)^*)^2\\
1 & f(z)^*
\end{pmatrix}.
\label{eq:dzP_from_dzf}
}
Since the matrix on the right-hand side is nonzero,
\eq{
\partial_zP_{\bsl{k}}=0
\quad\Longleftrightarrow\quad
\partial_zf(z)=0.
\label{eq:dzP_zero_dzf_zero}
}
At the pole $z=p$, we instead introduce $\widetilde{f}(z) \equiv \frac{1}{f(z)} = \frac{z-p}{c_p+v(z-p)}$ with $c_p\neq 0$ and $v(0)=0$, and use the corresponding expression for the
projector in \cref{eq:P_ftilde}.
The upper-right element of $\partial_zP_{\bsl{k}}$ is then $\left[\partial_zP_{\bsl{k}}\right]_{12}
= \partial_z\widetilde{f}(z)$, which is nonzero at $z=p$.
Therefore, the pole is not a zero of the quantum metric. 
Combining the above results, the zeros of $\Tr[g_{\boldsymbol{k}}]$ happens precisely when $\partial_zf(z)=0$.

We now just need to show that $\partial_z f$ has zeros.
Let $f'(z)\equiv\partial_z f(z)$.
Since the multipliers of $f$ after shifting reciprocal lattice vectors are constant,
\eq{ f'(z+1)=\eta_1f'(z),
\quad f'(z+\omega)=\eta_2f'(z).}
Consequently, the logarithmic derivative $\partial_z f'/f'$ is periodic with periods $1$ and $\omega$. Choose a fundamental parallelogram $\mathcal F$ whose boundary contains no zeros or poles of $f'$.
The contributions from opposite sides of $\partial\mathcal F$ cancel, so the argument principle gives
\eq{
N_{\mathrm z}-N_{\mathrm p}
= \frac{1}{2\pi\ii}
\oint_{\partial\mathcal F}
\frac{\partial_zf'(z)}{f'(z)}dz
= 0,
\label{eq:df_zero_pole_count}}
where the number of zeros ($N_{\mathrm z}$) and the number of poles ($N_{\mathrm p}$) are counted with multiplicity. Near the unique simple pole $p$ of $f$, $f(z)
= \frac{c_{p}+v(z-p)}{z-p}$, and hence
\eq{f'(z) = -\frac{c_{p}}{(z-p)^2}
+ O(1).}
Thus, $f'$ has one double pole and no other poles in a
fundamental domain, so $N_{\mathrm p}=2$. It follows from
\cref{eq:df_zero_pole_count} that $N_{\mathrm z}=2$. Therefore, $\partial_z f(z)$ has exactly two zeros, counted with multiplicity. So $\Tr[g_{\boldsymbol{k}}]$ vanishes at some
momenta.

\medskip

\section{Nonexistence of isolated and critical Chern-ideal bands in finite-band tight-binding models with finite-range hopping}
The finite-range hopping means that the hopping matrix $t_{\bsl{R}}$ in \cref{eq:h_from_t_general_embedding} is nonzero only for a finite number of $\bsl{R}$.
Here we have renamed the previous summation variable
$\bsl{\Delta R}$ as $\bsl{R}$ for simplicity.

\subsection{Isolated Chern-ideal bands}
\label{subsec:Isolated Chern Ideal Bands}

For this proof, we need to first rewrite the expression of the Hamiltonian.
For the hopping from the orbital $\alpha'$ to the orbital $\alpha$ with a lattice vector $\bsl{R}$, we define the embedding-dependent Fourier
displacement
\eq{
\bsl{\rho}
= \bsl{R} + \bsl{\xi}_{\alpha}
- \bsl{\xi}_{\alpha'}.
\label{eq:rho_general_embedding}
}
Finite bands and finite-range hopping together imply that only finitely many triples
$(\alpha,\alpha',\bsl{R})$ satisfy
$\left[t_{\bsl{R}}\right]_{\alpha\alpha'}\neq0$.
Consequently, the set
\eq{
\mathcal D
= \left\{
\bsl{R} + \bsl{\xi}_{\alpha} - \bsl{\xi}_{\alpha'}
\mid \bsl{R}\in\mathcal L, 
\alpha,\alpha'\in\{1,\ldots,N\},
\left[t_{\bsl{R}}\right]_{\alpha\alpha'}\neq0
\right\}
\label{eq:embedding_displacement_set}
}
is finite, where we recall $\mathcal{L}$ is the two-dimensional Bravais lattice. Different orbital pairs may give the same displacement $\bsl{\rho}$. We therefore group all terms with the same $\bsl{\rho}$ into one element in $\mathcal{D}$. For each $\bsl{\rho}\in\mathcal D$, define the $N\times N$ matrix $A_{\bsl{\rho}}$ entrywise by
\eq{\left[A_{\bsl{\rho}}\right]_{\alpha\alpha'} =
\begin{cases}
\left[t_{\bsl{R}}\right]_{\alpha\alpha'},
&
\text{if there exists }\bsl{R}\in\mathcal L \text{ such that } \bsl{R} + \bsl{\xi}_{\alpha} - \bsl{\xi}_{\alpha'} = \bsl{\rho},\\
0, 
&
\text{otherwise}.
\end{cases}
\label{eq:finite_hopping_matrix}}
The Bloch Hamiltonian can therefore be written as the finite Fourier sum
\eq{
h(\bsl{k}) = \sum_{\bsl{\rho}\in\mathcal D}
e^{-\ii\bsl{k}\cdot\bsl{\rho}}
A_{\bsl{\rho}}. \label{eq:general_embedding_finite_fourier_sum}
}
Notice that $\bsl{\rho}\in\mathcal{D}$ need not be a Bravais-lattice
vector. The proof below uses only that $\mathcal D$ is finite. The transformation rule of \cref{eq:general_embedding_finite_fourier_sum} under the shift of reciprocal lattice vectors is the same as  \cref{eq:generalEmbedding}.

\begin{theorem}
\label{thm:finite_range_isolated_no_go}
A finite-band $h(\bsl{k})$ with finite-range hopping cannot realize an isolated Chern-ideal band.
\end{theorem}

\begin{proof}
Assume, for contradiction,
that $h(\bsl{k})$ has an isolated Chern-ideal band
$E_{\bsl k}$ with $\Ch>0$.
Since the band is Chern-ideal with $\Ch>0$,
we have 
\eq{
\int_{\BZ}d^2k
\operatorname{Tr}g(\bsl{k})
=
2\pi\Ch
> 0\ ,
}
which means there exists $\bsl{k}_0\in\BZ$ such that $\operatorname{Tr}g(\bsl{k}_0)>0$.
By continuity, we can choose a small open \emph{square} neighborhood $\mathcal U $ of $\bsl{k}_0$, \ie, $\mathcal U = \{ (k_x , k_y) \mid k_x \in  (-2\epsilon+k_{0,x}, 2\epsilon+k_{0,x}), k_y \in (-2\epsilon+k_{0,y}, 2\epsilon+k_{0,y})\}$, for some small positive $\epsilon$, and $\operatorname{Tr}g(\bsl{k})>0$ on $\mathcal U$.
Recall that the ideality gives $\left(\partial_{k_x}-i\partial_{k_y}\right)U_{\bsl{k}}
=\lambda_{\bsl{k}} U_{\bsl{k}}$ [\cref{eq:partial_z_U}] on $\mathcal U$.
Since $U_{\bsl{k}}$ is normalized, we can pick $\mathcal U$ small enough such that one component $U_{\bsl{k},a}$ is everywhere nonzero on $\mathcal U$.
Then, we have
\eq{
\left(\partial_{k_x}-\ii\partial_{k_y}\right)\frac{U_{\bsl{k}}}{U_{\bsl{k},a}} = 0\ ,
}
which means 
\eq{
\frac{U_{\bsl{k}}}{U_{\bsl{k},a}} = w(k)
}
with $k=k_x - \ii k_y$ is holomorphic on $\mathcal{U}$.
Explicitly, $w(k)$ has the form
\eq{
w(k)=
\begin{pmatrix}
w_1(k)\\
\vdots\\
w_{a-1}(k)\\
1\\
w_{a+1}(k)\\
\vdots\\
w_N(k)
\end{pmatrix},
\label{eq:Nband_w_form}
}
where all components $w_{\alpha}(k)$ are holomorphic on the corresponding neighborhood of $\mathcal U$, and $w_a(k)=1$.
Note that here and from now on, we use $\mathcal U$ to denote both the neighborhood of $(k_x,k_y)$ in $\mathbb{R}^2$, as well as the neighborhood of $k_x-\ii k_y$ in $\mathbb{C}$, for simplicity.
Because $w(k)$ differs from $U$ only by a nonzero scalar factor, it is still an eigenvector of $h(\bsl k)$:
\eq{
h(\bsl k) w(k)=E_{\bsl k} w(k).
\label{eq:Nband_Hw_Ew}
}

Let $S_a=\{1,\ldots,N\}\setminus\{a\}.$ Since $w_a(k)=1$, the $a$-th component of
\cref{eq:Nband_Hw_Ew} gives
\eq{
[h(\bsl k)w(k)]_a=E_{\bsl k}.
\label{eq:Nband_a_component}
}
For $\alpha \in S_a$, the $\alpha$-th component gives
\eq{
[h(\bsl k)w(k)]_{\alpha}=E_{\bsl k}w_{\alpha}(k).
\label{eq:Nband_alpha_component}
}
Eliminating $E_{\bsl k}$ between \cref{eq:Nband_a_component} and
\cref{eq:Nband_alpha_component}, we obtain
\eq{
[h(\bsl k)w(k)]_{\alpha}
-
w_{\alpha}(k)[h(\bsl k)w(k)]_a
=
0,
\qquad
\alpha \in S_a.
\label{eq:Nband_eliminate_E}
}
Now insert the finite-range expansion $h(\bsl k)=\sum_{\bsl{\rho}\in\mathcal{D}}e^{-i\bsl \rho\cdot\bsl k}A_{\bsl \rho}.$
Then \cref{eq:Nband_eliminate_E} becomes
\eq{
\sum_{\bsl{\rho}\in\mathcal{D}}e^{-i\bsl \rho\cdot\bsl k}
\left(
[A_{\bsl \rho}w(k)]_{\alpha}
- w_{\alpha}(k)[A_{\bsl \rho}w(k)]_a
\right)
= 0,
\qquad
\alpha \in S_a.
\label{eq:Nband_sum_CR}
}
Define
\eq{
C_{\alpha,\bsl \rho}(k)
\equiv
[A_{\bsl \rho}w(k)]_{\alpha}
-
w_{\alpha}(k)[A_{\bsl \rho}w(k)]_a,
\qquad
\alpha \in S_a.
\label{eq:Nband_C_def}
}
Since $w(k)$ is holomorphic on $\mathcal{U}$ and $A_{\bsl \rho}$ is constant, each
$C_{\alpha,\bsl \rho}(k)$ is holomorphic on $\mathcal{U}$.
\cref{eq:Nband_sum_CR} becomes
\eq{
\sum_{\bsl{\rho}\in\mathcal{D}}e^{-i\bsl \rho \cdot\bsl k}C_{\alpha,\bsl \rho}(k)=0,
\qquad
\alpha \in S_a.
\label{eq:Nband_sum_C}
}

We want to show each element of the summation in \cref{eq:Nband_sum_C} is zero.
First, since $C_{\alpha,\bsl \rho}(k)$ is a holomorphic function of $k$ on $\mathcal{U}$, it admits a convergent power series expansion about any $k' \in \mathcal{U}$,
\eqa{
C_{\alpha,\bsl{\rho}}(k) & = \sum_{n=0}^{\infty}c_{\alpha,\bsl{\rho},n}(k-k')^n \\
& = \sum_{n=0}^{\infty}c_{\alpha,\bsl{\rho},n}\bigl[ (k_x-k_{x}') - \ii(k_y-k_{y}') \bigr]^n \\
& = \sum_{n=0}^{\infty}c_{\alpha,\bsl{\rho},n} \sum_{m=0}^{n} \binom{n}{m} (-\ii)^{n-m} (k_x-k_{x}')^m (k_y-k_{y}')^{n-m}.}
Thus, $C_{\alpha,\bsl \rho}(k)$ is a real-analytic function of $(k_x,k_y)$ on $\mathcal U$.
Since the factors $e^{-i\bsl \rho \cdot\bsl k}$ are also real-analytic in
$(k_x,k_y)$, and the sum over $\bsl \rho \in \mathcal{D}$ is finite, the left-hand
side of \cref{eq:Nband_sum_C} is real-analytic on $\mathcal U$, \ie, the following function is real-analytic on $\mathcal U $:
\eq{
\Gamma_\alpha(k_x,k_y) = \sum_{\bsl{\rho}\in\mathcal{D}}e^{-i\bsl \rho \cdot\bsl k}C_{\alpha,\bsl \rho}(k)
}
The real-analytical allows a holomorphic extension \eq{
\Gamma_\alpha(K_x,K_y) = \sum_{\bsl{\rho}\in\mathcal{D}}e^{-i \rho_x K_x}  e^{-i \rho_y K_y} C_{\alpha,\bsl \rho}(K_x - \ii K_y)
}
with $(K_x,K_y)$ in an open subset of $\dsC^2$ that contains $\mathcal{U}$.
Then, there exists a small enough $\mu>0$ such that 
\eq{
\mathcal{W} = \{(k_x+\ii \kappa_x,k_y+\ii \kappa_y)|k_x\in I_x \equiv (k_{0,x}-\epsilon,k_{0,x}+\epsilon), k_y\in I_y \equiv (k_{0,y}-\epsilon,k_{0,y}+\epsilon), \kappa_x \in (-\mu,\mu) , \kappa_y \in (-\mu,\mu) \}
\label{eq:KxKy}
}
is contained in that open subset---$\Gamma_\alpha(K_x,K_y)$ is holomorphic on the open $\mathcal{W}\subset \dsC^2$.
Since $\Gamma_\alpha(k_x,k_y) = 0$ on the open $\mathcal{U}\subset \dsR^2$, the identity theorem implies
\eq{
\label{eq:Nband_complexified_sum}
\Gamma_\alpha(K_x,K_y) = \sum_{\bsl{\rho}\in\mathcal{D}}e^{-i \rho_x K_x}  e^{-i \rho_y K_y} C_{\alpha,\bsl \rho}(K_x - \ii K_y) = 0
}
for $(K_x,K_y)\in \mathcal{W}\subset \dsC^2$.
To see that, first note that $\Gamma_\alpha(K_x,k_y)$ is a holomorphic function of $K_x \in I_x + \ii (-\mu,\mu)\subset \dsC$.
Then, since $\Gamma_\alpha(k_x,k_y) = 0$ for all $k_x\in I_x$, the one-variable identity theorem implies $\Gamma_\alpha(K_x,k_y) = 0$ for all $K_x \in I_x + \ii (-\mu,\mu)$, as $I_x$ certainly contains accumulation points.
Then, we know $\Gamma_\alpha(K_x,k_y) = 0 $ for all $K_x \in I_x + \ii (-\mu,\mu)$ and all $k_y\in I_y$.
Now consider a fixed $K_x \in I_x + \ii (-\mu,\mu)$, and then $\Gamma_\alpha(K_x,K_y)$ is a holomorphic function of $K_y\in I_y + \ii (-\mu,\mu)$, which is also an open subset of $\dsC$.
Then, for the fixed $K_x$, the one-variable identity theorem implies  $\Gamma_\alpha(K_x,K_y) = 0$ for all $K_y\in I_y + \ii (-\mu,\mu)$, since $\Gamma_\alpha(K_x,k_y) = 0$ for $k_y\in I_y$ and $I_y$ contains accumulation points.
Since the choice of the fixed $K_x$ is arbitrary, we know 
$\Gamma_\alpha(K_x,K_y) = 0$ for all $(K_x,K_y)\in \mathcal{W}$, thus
\cref{eq:Nband_complexified_sum}.

An alternative way to see this is the following.
First we know $\Gamma_\alpha(k_x,k_y)$ is real-analytic near any $\bsl k'\in \mathcal{U}$----it has a
convergent expansion $\sum_{p,q\ge0}
c_{pq}
(k_x-k_{x}')^p
(k_y-k_{y}')^q.$
But this function is zero on $\mathcal{U}$, so all real
partial derivatives at $\bsl k'$ vanish. Hence the coefficients $c_{pq}
=
\frac{1}{p!q!}
\left[\partial_{k_x}^p\partial_{k_y}^q
\Gamma_\alpha(k_x,k_y)
\right]_{\bsl k=\bsl k'}
=
0$. Therefore, the holomorphic extension in $K_x$ and $K_y$, 
$\Gamma_\alpha(K_x,K_y)$, 
which can be obtained by the same power series, also vanishes in a neighborhood of $\bsl k'$. Therefore,
\cref{eq:Nband_complexified_sum} is true.

Next, express $\bsl \rho\cdot\bsl k$ in the coordinates $k,{k^*}$ and $r_{\bsl{\rho}},r^*_{\bsl{\rho}}$ 
\eq{
k=k_x-\ii k_y,
\qquad
{k^*}=k_x+\ii k_y,
\qquad
r_{\bsl{\rho}} = \frac{\rho_x+\ii\rho_y}{2},
\qquad
r^*_{\bsl{\rho}} = \frac{\rho_x-\ii\rho_y}{2}.
}
Then,
\eq{
\bsl \rho\cdot\bsl k =
\rho_xk_x+\rho_yk_y =
\frac{\rho_x+\ii\rho_y}{2} k
+
\frac{\rho_x-\ii\rho_y}{2} {k^*} =
r_{\bsl{\rho}}k
+
r^*_{\bsl{\rho}}{k^*}.
}
Therefore, \cref{eq:Nband_sum_C} becomes
\eq{
\sum_{\bsl{\rho}\in\mathcal{D}}
e^{-\ii r_{\bsl{\rho}}k}
e^{-\ii r^*_{\bsl{\rho}}{k^*}}
C_{\alpha,\bsl \rho}(k)
=
0,
\qquad
\alpha\in S_a.
\label{eq:Nband_sum_k_kbar}
}
Furthermore, recall \cref{eq:KxKy}, and now define
\eq{
Y=K_x-\ii K_y = k_x + \kappa_y- \ii k_y + \ii \kappa_x ,
\qquad
W=K_x+\ii K_y =  k_x - \kappa_y + \ii k_y + \ii \kappa_x.
\label{eq:Nband_Y_W}
}
Then \cref{eq:Nband_sum_k_kbar} is
\eq{
\sum_{\bsl{\rho}\in\mathcal{D}}
e^{-ir_{\bsl{\rho}}Y}
e^{-ir^*_{\bsl{\rho}}W}
C_{\alpha,\bsl \rho}(Y)
=
0,
\qquad
\alpha \in S_a, (K_x,K_y)\in \mathcal{W}\ .
}
Note that as $K_x$ and $K_y$ are independent, $Y$ and $W$ are independent complex variables.

In particular, for any $k'=k_x'-\ii k_y'$ with $(k_x',k_y')\in I_x \times I_y$, and fix $Y=k'$ we have 
\eq{
\sum_{\bsl{\rho}\in\mathcal{D}}
\left[
e^{-ir_{\bsl{\rho}}k'}C_{\alpha,\bsl \rho}(k')
\right]
e^{-ir^*_{\bsl{\rho}}W} = 0
\label{eq:Nband_exponential_sum}
}
that holds for any value of $W = k_x'  + \ii k_y' - 2 \kappa_y+2\ii \kappa_x$ in $ (k_x'-2 \mu'_x, k_x'+2\mu'_x)  + \ii (k_y'-2 \mu'_y, k_y'+2\mu'_y)$, where
\eqa{
& \mu'_x = \min(k_y'+\epsilon-k_{0,y},k_{0,y}+\epsilon-k_y',\mu) > 0  \\
& \mu'_y = \min(k_x'+\epsilon-k_{0,x},k_{0,x}+\epsilon-k_x',\mu) > 0  \ .
}
The $r^*_{\bsl{\rho}}$'s are distinct for distinct $\bsl \rho$. Since a finite number of exponential functions of $W$ with
distinct$\bsl \rho$ are linearly independent, every coefficient in
\cref{eq:Nband_exponential_sum} must vanish, leading to 
\eq{e^{-ir_{\bsl{\rho}}k'}C_{\alpha,\bsl \rho}(k')=0\ ,\forall k'\in I_x + \ii I_y, \alpha \in S_a, \bsl{\rho} \in \mathcal{D}\ .
\label{eq:coefficient_vanish}
}

Let us show \cref{eq:coefficient_vanish} explicitly. Write $\mathcal{D}=\{\bsl \rho_1,\ldots,\bsl \rho_n\},
\quad
\gamma_j(k')\equiv e^{-i r_{\bsl \rho_j}k'}C_{\alpha,\bsl \rho_j}(k'),
\quad
r^*_j\equiv r^*_{\bsl \rho_j}.$
For fixed $k'$, \cref{eq:Nband_exponential_sum} becomes
\eq{
G_{k'}(W)\equiv \sum_{j=1}^n \gamma_j(k')e^{-i r^*_j W}=0.
\label{eq:Nband_exp_sum_fixed_Y}
}
Since $G_{k'}(W)$ vanishes identically as a local holomorphic function of $W$ in $ (k_x'-2 \mu'_x, k_x'+2\mu'_x)  + \ii (k_y'-2 \mu'_y, k_y'+2\mu'_y)$, all of its derivatives vanish. Taking $\ell=0,1,\ldots,n-1$ derivatives with respect to $W$ and evaluating at $W_0=k_{x}'+\ii k_{y}'=k'^*$, we obtain
\eq{
0=
\left.\frac{d^\ell G_{k'}}{dW^\ell}\right|_{W=W_0}
=
\sum_{j=1}^n (-i r^*_j)^\ell \gamma_j(k') e^{-ir^*_{j}k'^*},
\qquad
\ell=0,1,\ldots,n-1.
\label{eq:Nband_vandermonde_system}
}
Equivalently,
\eq{ 
\begin{pmatrix}
1 & 1 & \cdots & 1\\
-\ii r^*_1 & -\ii r^*_2 & \cdots & -\ii r^*_n\\
(-\ii r^*_1)^2 & (-\ii r^*_2)^2 & \cdots & (-\ii r^*_n)^2\\
\vdots & \vdots & \ddots & \vdots\\
(-\ii r^*_1)^{n-1} & (-\ii r^*_2)^{n-1} & \cdots &
(-\ii r^*_n)^{n-1}
\end{pmatrix}
\begin{pmatrix}
\gamma_1(k')e^{-\ii r_1^*k'^*}\\
\gamma_2(k')e^{-\ii r_2^*k'^*}\\
\gamma_3(k')e^{-\ii r_3^*k'^*}\\
\vdots\\
\gamma_n(k')e^{-\ii r_n^*k'^*}
\end{pmatrix}
=
\begin{pmatrix}
0\\
0\\
0\\
\vdots\\
0
\end{pmatrix}.
\label{eq:Nband_vandermonde_matrix}
}
The matrix in \cref{eq:Nband_vandermonde_matrix} is a Vandermonde matrix, whose determinant is $\prod_{1\le p<q\le n}
\left[
(-\ii r^*_q)-(-\ii r^*_p)
\right]
\neq0$, as the $ r^*_j$'s are distinct. Therefore, the only solution is $\gamma_j(k')e^{-\ii r_j^*k'^*}=0$ for all $j=1,\ldots,n.$ Since $e^{-\ii r_j^*k'^*}\neq0$, this implies
$\gamma_j(k')=0$ for all $j=1,\ldots,n$.
Namely,
$e^{-ir_{\bsl{\rho}}k'}C_{\alpha,\bsl \rho}(k')=0,$ which is exactly \cref{eq:coefficient_vanish}.
Note that we use the finite property of  $\mathcal{D}$ in this step. For exponentially decaying hopping, $\mathcal{D}$ has infinite elements, so the above argument is not necessarily generalizable to guarantee $e^{-ir_{\bsl{\rho}}k}C_{\alpha,\bsl \rho}(k)$ is zero for each $\bsl{\rho}$ in this case.

Using the definition of $C_{\alpha,\bsl \rho}$ in \cref{eq:Nband_C_def}, this means
\eq{
[A_{\bsl \rho}w(k)]_{\alpha}
=
w_{\alpha}(k)[A_{\bsl \rho}w(k)]_a,
\qquad
\alpha \in S_a.
\label{eq:Nband_component_relation_AR}
}
Define $\lambda_{\bsl \rho}(k)\equiv [A_{\bsl \rho}w(k)]_a.$ Since $w_a(k)=1$, the same relation also holds for $\alpha=a$. Therefore, for every component $\alpha=1,\ldots,N$,
\eq{
[A_{\bsl \rho}w(k)]_{\alpha}
=
\lambda_{\bsl \rho}(k)w_{\alpha}(k).
\label{eq:Nband_all_components}
}
Equivalently,
\eq{
A_{\bsl \rho}w(k)=\lambda_{\bsl \rho}(k)w(k),
\qquad
\forall \bsl \rho \in \mathcal{D}.
\label{eq:Nband_key_ARw}
}
In addition, \cref{eq:Nband_a_component} shows that
\eq{
\label{eq:E_k_lambda}
E_{\bsl{k}} = [h(\bsl k)w(k)]_a =  \sum_{\bsl{\rho}\in\mathcal{D}}e^{-i \bsl \rho\cdot \bsl k}  \lambda_{\bsl \rho}(k) [w(k)]_a = \sum_{\bsl{\rho}\in\mathcal{D}}e^{-i \bsl \rho\cdot \bsl k}  \lambda_{\bsl \rho}(k)\ .
}

Now we use the fact that the number of bands $N$ is finite, which means that $A_{\bsl\rho}$ is an $N\times N$ matrix.
The $a$-th component of \cref{eq:Nband_key_ARw} gives
\eq{
\lambda_{\bsl\rho}(k) =
[A_{\bsl\rho}w(k)]_a
=
\sum_{\alpha'=1}^{N}
[A_{\bsl\rho}]_{a\alpha'}
w_{\alpha'}(k)
=
[A_{\bsl\rho}]_{aa}
+
\sum_{\alpha'\neq a}
[A_{\bsl\rho}]_{a\alpha'}
w_{\alpha'}(k),
\label{eq:Nband_lambda_holomorphic}
}
where $w_a(k)=1$.
Thus, $\lambda_{\bsl \rho}(k)$ is locally holomorphic  in $k$. Moreover,
\cref{eq:Nband_key_ARw} says that $\lambda_{\bsl \rho}(k)$ is an eigenvalue of the
fixed $N\times N$ matrix $A_{\bsl \rho}$. A fixed $N\times N$ matrix has at most $N$ eigenvalues, so $\lambda_{\bsl \rho}(k)$ takes values in a finite set. Since $\lambda_{\bsl \rho}(k)$ is holomorphic on the connected neighborhood $\mathcal{U}$ and takes values in
a finite set, it must be constant. Hence, $\lambda_{\bsl \rho}(k)=\lambda_{\bsl \rho}$ is independent of $k$ and
\eq{
A_{\bsl \rho} w(k)=\lambda_{\bsl \rho} w(k),
\qquad
\forall k\in \mathcal{U},
\quad
\forall \bsl \rho \in \mathcal{D}.
\label{eq:Nband_ARw_constant_lambda}
}
We now show that $w(k)$ also has to be constant.
For contradiction, assume that $w(k)$ is nonconstant
on $\mathcal U$.
Then there exist $k_1,k_2\in \mathcal U$ such that $w(k_1)\neq w(k_2)$. since $w_a(k_1)=w_a(k_2)=1$, $w(k_1)$ and $w(k_2)$ are linearly
independent.
Define $L=\operatorname{span}\{w(k) \mid k\in \mathcal U\}\subseteq \mathbb C^N.$
Since $\dim\operatorname{span}\{w(k_1),w(k_2)\}=2$, then $\dim L\ge2.$
\Cref{eq:Nband_ARw_constant_lambda} and linearity imply
\eq{
A_{\bsl \rho} v=\lambda_{\bsl \rho} v,
\qquad
\forall v\in L,
\quad
\forall \bsl \rho \in \mathcal{D}.
\label{eq:Nband_AR_on_L}
}
Thus, $L$ is a common eigensubspace, or simultaneous eigensubspace, of all hopping
matrices $A_{\bsl \rho}$.
Now apply $h(\bsl k)$ to any vector $v\in L$:
\eq{
h(\bsl k) v
=
\sum_{\bsl{\rho}\in\mathcal{D}}e^{-i \bsl \rho\cdot \bsl k}A_{\bsl \rho} v.
\label{eq:Nband_H_on_L_1}
}
Using \cref{eq:Nband_AR_on_L},
\eq{
h(\bsl k) v
=
\sum_{\bsl{\rho}\in\mathcal{D}}e^{-i \bsl \rho \cdot \bsl k}\lambda_{\bsl \rho} v.
\label{eq:Nband_H_on_L_2}
}
Combined with \cref{eq:E_k_lambda}, we have
\eq{
h(\bsl k)v=E_{\bsl k}v,
\qquad
\forall v\in L.
\label{eq:Nband_H_degenerate}
}
Since $\dim L\ge2$, \cref{eq:Nband_H_degenerate} implies that $h(\bsl k)$ has an
eigenspace of dimension at least two at energy $E_{\bsl k}$, contradicting that the target band is a single isolated band. 
Therefore, $w(k)$ must be constant on $\mathcal U$.

However, a constant $w(k)$ would make
$P_{\bsl k}$ constant on $\mathcal U$, and therefore $\operatorname{Tr}g(\bsl{k}_0)=0$.
This contradicts the choice
$\operatorname{Tr}g(\bsl{k}_0)>0$. 

In conclusion, no finite-range
finite-band tight-binding Hamiltonian can realize an isolated
Chern-ideal band with $\Ch>0$.
The same conclusion holds for $\Ch<0$ which can be derived by applying a time-reversal operation to the above derivations.
\end{proof}

\subsection{Critical Chern-ideal bands}

The proof in \cref{subsec:Isolated Chern Ideal Bands} assumes that the target band is isolated over the BZ.
This global isolation assumption guarantees that the rank-one projector is real-analytic everywhere over the BZ.
We now allow band touchings.
In this case, the projector may remain continuous and has finite Berry curvature over the BZ, while being real-analytic away from the touching points.
This is the critical topological band proposed in \refcite{li2026stabletopologyexactlyflat}, which was demonstrated to exist in finite-band finite-range tight-binding model.
We now show that even such critical topological band cannot be Chern ideal with $|\Ch|>0$ in finite-band finite-range tight-binding models.

\begin{theorem}
\label{thm:finite_range_critical_no_go}
A finite-band tight-binding Hamiltonian with finite-range hoppings cannot realize a critical Chern-ideal band with $|\Ch|>0$, as long as the band touching points form a measure-zero set in $\BZ$ and the Berry curvature is finite over $\BZ$.
\end{theorem}

\begin{proof}

Since time-reversal changes the sign of the Chern number, we consider the case $\Ch>0$ without loss of generality.
Let $\Sigma\subset \BZ$ be the set of $\bsl{k}$-points where band touching occurs.
Then define
\eq{
M=\BZ \setminus \Sigma
\label{eq:critical_M_def}
}
and $E(\bsl{k})$ is nondegenerate for $\bsl{k}\in M$.
Since a finite-range hopping Bloch Hamiltonian is real-analytic,
\Cref{thm:P_smooth_at_Isolated}
in \Cref{app: basics} imply that
$P_{\bsl k}$ is real-analytic, and hence smooth, on $M$.
Since $\Ch>0$, the Berry curvature is finite for all $\bsl{k}$ and $\Sigma$ has measure zero, there must exist a 
$\bsl k_0\in M$ such that
\eq{
\operatorname{Tr}g(\bsl k_0)=F_{\bsl{k_0}}>0.
}
In a small enough neighborhood of $\bsl{k}_0$, the ideal condition implies that we can construct a local holomorphic eigenvector $w(k)$ as in \cref{eq:w_k}.
Then, by the same logic of the discussion following \cref{eq:Nband_ARw_constant_lambda},
either $w(k)$ is constant in the neighborhood and hence $\operatorname{Tr}g(\bsl k_0)=0$, or $E(\bsl{k})$ is degenerate in the neighborhood.
This leads to a contradiction.
Therefore, a Hamiltonian with finite-range hoppings cannot realize a critical Chern-ideal band with $|\Ch|>0$.
\end{proof}
\end{document}